\documentclass[journal,twoside,draftcls,onecolumn]{IEEEtran} % SINGLE COLUMN
\usepackage{cite}
\usepackage[pdftex]{graphicx}
\graphicspath{{Figure/}}
\usepackage[cmex10]{amsmath}
\usepackage{amsthm}
\usepackage{amssymb}
\usepackage{algorithmic}
\usepackage{array}
\usepackage{color}
\usepackage{epstopdf}
\usepackage{amsfonts}

\usepackage{pgfplots}
\pgfplotsset{compat=1.3}
\usepackage{tikz}							% Tikz
\usetikzlibrary{shapes}
\usetikzlibrary{spy}
\usetikzlibrary{circuits}

\newcommand{\mat}[1]{\mathbf {#1}}
\newcommand{\vect}[1]{{\mathbf #1}}
\newcommand{\tr}{\text{tr}}
\newcommand{\MSE}{\mathrm{MSE}}
\newcommand{\HDL}{\mat{H}_\mathrm{DL}}

\newcommand{\fref}[1]{Fig.~\ref{#1}}

\newcommand{\eref}[1]{(\ref{#1})}
\newcommand{\cref}[1]{Chapter~\ref{#1}}
\newcommand{\sref}[1]{Section~\ref{#1}}

\newtheorem{theorem}{Theorem}[section]
\newtheorem{proposition}[theorem]{Proposition}
\newtheorem{lemma}{Lemma}
\newtheorem{corollary}{Corollary}

\ifCLASSOPTIONcompsoc
  \usepackage[caption=false,font=normalsize,labelfont=sf,textfont=sf]{subfig}
\else
  \usepackage[caption=false,font=footnotesize]{subfig}
\fi

\usepackage{dblfloatfix}
\begin{document}
%
% paper title
% Titles are generally capitalized except for words such as a, an, and, as,
% at, but, by, for, in, nor, of, on, or, the, to and up, which are usually
% not capitalized unless they are the first or last word of the title.
% Linebreaks \\ can be used within to get better formatting as desired.
% Do not put math or special symbols in the title.
\title{Single-Tap Precoders and Decoders for Multi-User MIMO FBMC-OQAM under Strong Channel Frequency Selectivity}
%
%
% author names and IEEE memberships
% note positions of commas and nonbreaking spaces ( ~ ) LaTeX will not break
% a structure at a ~ so this keeps an author's name from being broken across
% two lines.
% use \thanks{} to gain access to the first footnote area
% a separate \thanks must be used for each paragraph as LaTeX2e's \thanks
% was not built to handle multiple paragraphs
%

\author{Fran\c{c}ois Rottenberg,~\IEEEmembership{Student Member,~IEEE,}
        Xavier~Mestre,~\IEEEmembership{Senior Member,~IEEE,}
        Fran\c{c}ois Horlin,~\IEEEmembership{Member,~IEEE,}
        and~J\'er\^ome Louveaux,~\IEEEmembership{Member,~IEEE}% <-this % stops a space
%\thanks{J. Doe and J. Doe are with Anonymous University.}% <-this % stops a space
\thanks{Fran\c{c}ois Rottenberg and J\'er\^ome Louveaux are with the Universit\'e catholique de Louvain, 1348 Louvain-la-Neuve, Belgium (e-mail: francois.rottenberg@uclouvain.be;jerome.louveaux@uclouvain.be).}
\thanks{Xavier Mestre is with the Centre Tecnològic de Telecomunicacions
	de Catalunya, 08860 Barcelona, Spain (e-mail: xavier.mestre@cttc.cat).}%
\thanks{Fran\c{c}ois Rottenberg and Fran\c{c}ois Horlin are with the Universit\'e libre de Bruxelles, 1050 Brussel, Belgium (e-mail: fhorlin@ulb.ac.be).}
}
\maketitle

% As a general rule, do not put math, special symbols or citations
% in the abstract or keywords.
\begin{abstract}
The design of linear precoders or decoders for multi-user (MU) multiple-input multiple-output (MIMO) filterbank multicarrier (FBMC) modulations in the case of strong channel frequency selectivity is presented. {\color{black} The users and the base station (BS) communicate using space division multiple access (SDMA).} The low complexity proposed solution is based on a single tap per-subcarrier precoding/decoding matrix at the base station (BS) in the downlink/uplink. As opposed to classical approaches that assume flat channel frequency selectivity at the subcarrier level, the BS does not make this assumption and takes into account the distortion caused by channel frequency selectivity. The expression of the FBMC asymptotic mean squared error (MSE) in the case of strong channel selectivity derived in earlier works is developed and extended. The linear precoders and decoders are found by optimizing the MSE formula under two design criteria, namely zero forcing (ZF) or minimum mean squared error (MMSE). Finally, simulation results demonstrate the performance of the optimized design. As long as the number of BS antennas is larger than the number of users, it is shown that those extra degrees of freedom can be used to compensate for the channel frequency selectivity.
\end{abstract}

% Note that keywords are not normally used for peerreview papers.

% For peer review papers, you can put extra information on the cover
% page as needed:
% \ifCLASSOPTIONpeerreview
% \begin{center} \bfseries EDICS Category: 3-BBND \end{center}
% \fi
%
% For peerreview papers, this IEEEtran command inserts a page break and
% creates the second title. It will be ignored for other modes.
\IEEEpeerreviewmaketitle

\section{Introduction}

% The very first letter is a 2 line initial drop letter followed
% by the rest of the first word in caps.
% 
% form to use if the first word consists of a single letter:
% \IEEEPARstart{A}{demo} file is ....
% 
% form to use if you need the single drop letter followed by
% normal text (unknown if ever used by IEEE):
% \IEEEPARstart{A}{}demo file is ....
% 
% Some journals put the first two words in caps:
% \IEEEPARstart{T}{his demo} file is ....
% 
% Here we have the typical use of a "T" for an initial drop letter
% and "HIS" in caps to complete the first word.
%\IEEEPARstart{T}{his} demo file is intended to serve as a ``starter file''
%for IEEE journal papers produced under \LaTeX\ using
%IEEEtran.cls version 1.8a and later.
% You must have at least 2 lines in the paragraph with the drop letter
% (should never be an issue)
%%5\textsuperscript{th} generation wireless communication systems (5G) is expected for beginning of 2020. The standards for 5G are not yet defined. However, the broad variety of applications targeted by 5G translate in clear requirements for the design of its physical layer. Very high data rates together reveal the need for a modulation scheme with large spectral efficiency. Furthermore, the amount of interconnected devices together with the actual coarse fragmented spectrum will require a system with high flexibility regarding resource allocation and synchronization issues.

\IEEEPARstart{O}{rthogonal} frequency division multiplexing (OFDM) is the most popular multicarrier modulation scheme nowadays. It is used for instance in systems such as WiFi, long term evolution (LTE) or digital video broadcasting (DVB). OFDM has been very attractive mainly because of its low complexity of implementation. The introduction of the cyclic prefix (CP) in OFDM allows for easy channel equalization. Extension to multiple-input multiple-output (MIMO) scenarios is straightforward thanks to the OFDM orthogonality ensured in the complex domain. At the same time, due to the rectangular pulse shaping of the fast Fourier transform (FFT) filters, OFDM systems exhibit very high frequency leakage and poor stopband attenuation. Furthermore, the use of the CP in OFDM significantly reduces the spectral efficiency of the system.

In the light of the shortcomings of OFDM, Offset-QAM-based filterbank multicarrier (FBMC-OQAM) modulation has been regarded as an attractive alternative. Rather than using a rectangle pulse in time, FBMC-OQAM uses a pulse shape which is more spread out in time and has much larger stopband attenuation \cite{farhang2011ofdm}. This in turn translates into higher spectral efficiency and relaxed synchronization constraints \cite{FBMC_synchro}. Moreover, it does not require CP overhead, which allows for a larger spectral efficiency. These advantages come at the expense of an increase in the system implementation complexity.

\begin{figure}[t!]
	\centering
	\includegraphics[width=0.45\textwidth]{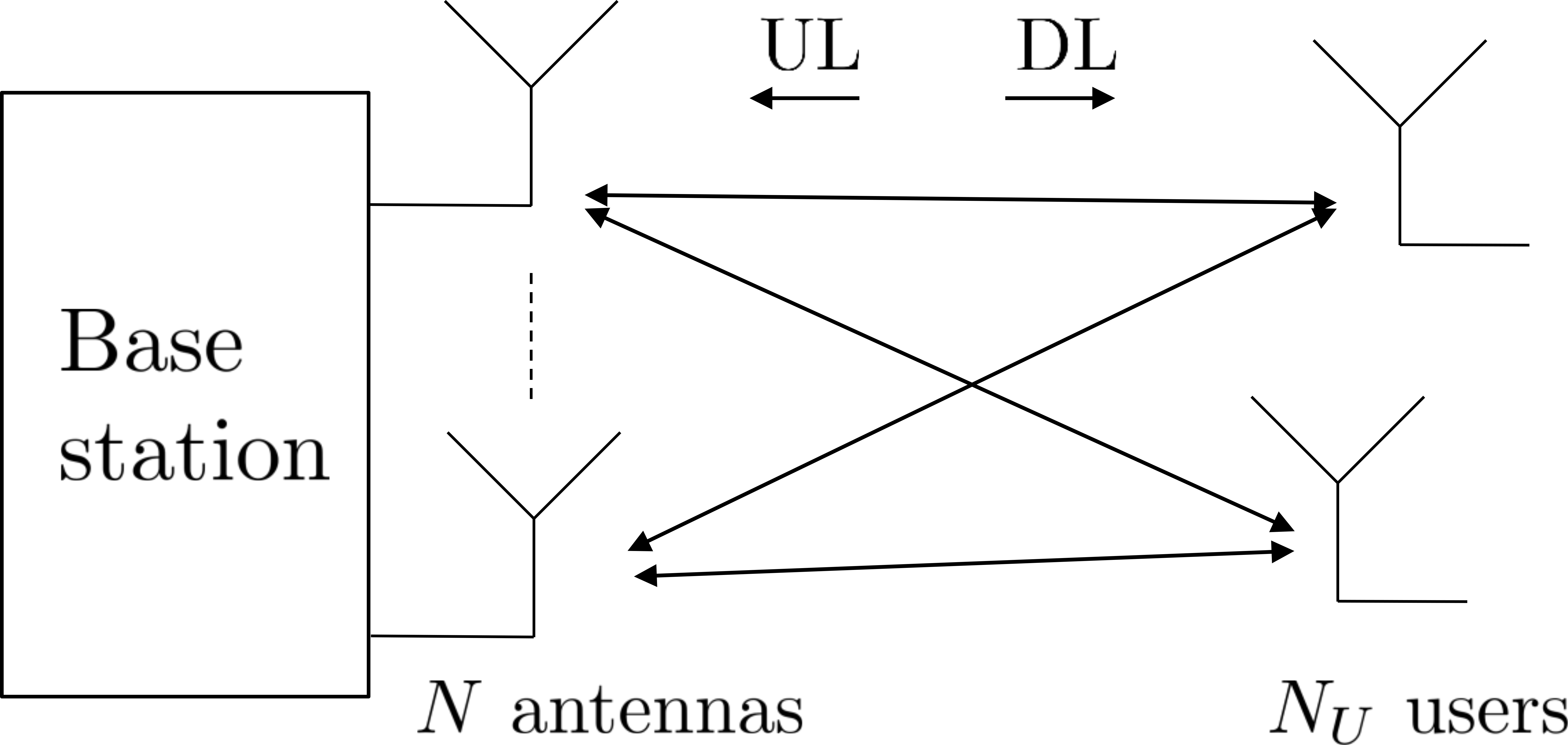}
	\caption{Multi-user MIMO scenario: $N_U$ single-antenna users and one base station with $N$ antennas communicate simultaneously in UL and DL using space division multiple access.}
	\label{fig:syst_model}
\end{figure}
Under frequency selective channels, single tap equalization is sufficient in OFDM to restore the system orthogonality. The same result occurs in FBMC if the assumption of a frequency flat channel at the subcarrier level is made, which is typically verified for mildly frequency selective channels. However, as the selectivity of the channel increases, FBMC begins to suffer from {\color{black}inter-symbol interference (ISI) and inter-carrier interference (ICI)} and the orthogonality is progressively destroyed \cite{mestre13tsp,ihalainen2007channel}. Many works in the literature have investigated this problem in the SISO case \cite{ihalainen2007channel,waldhauser2008mmse,baltar2009mmse,ikhlef2009enhanced} and later on for the MIMO case, see \cite{7491375} for recent review paper on the subject. Most of the approaches to mitigate channel frequency selectivity are based on the design of multi-tap fractionally spaced equalizers. For instance, in \cite{ihalainen2011channel}, the authors designed multi-tap decoding matrices following a frequency sampling design, i.e., they compute the time domain equalizer coefficients so that its frequency response passes through some well chosen target frequency points. On the other hand, \cite{caus2012transmitter} proposes a multi-tap filtering solution at both transmit and receive sides. {\color{black} This problem has also been analyzed in the MU MIMO context in several works. In \cite{6612142}, the authors extend the block diagonalization technique to FBMC systems. Through an iterative algorithm, the work of \cite{6877914} alleviates the dimensionality constraint of \cite{6612142} by allowing designs where the total number of receive antennas of the users exceeds the number of transmit antennas at the base station. In \cite{7178407,7499163}, multi-tap precoders and decoders are iteratively and jointly designed.} Moreover, the work originally devised for the SISO case in \cite{mestre13tsp} and later extended for the MIMO case in \cite{mestre2014parallel,Mestre16} proposes instead a parallel multi-stage processing architecture at both sides of the communication link. One should however note that iterative designs, multi-tap filtering and multi-stage processing increase the complexity of the system.

\begin{figure}[t!]
	\centering
		\includegraphics[width=0.7\textwidth]{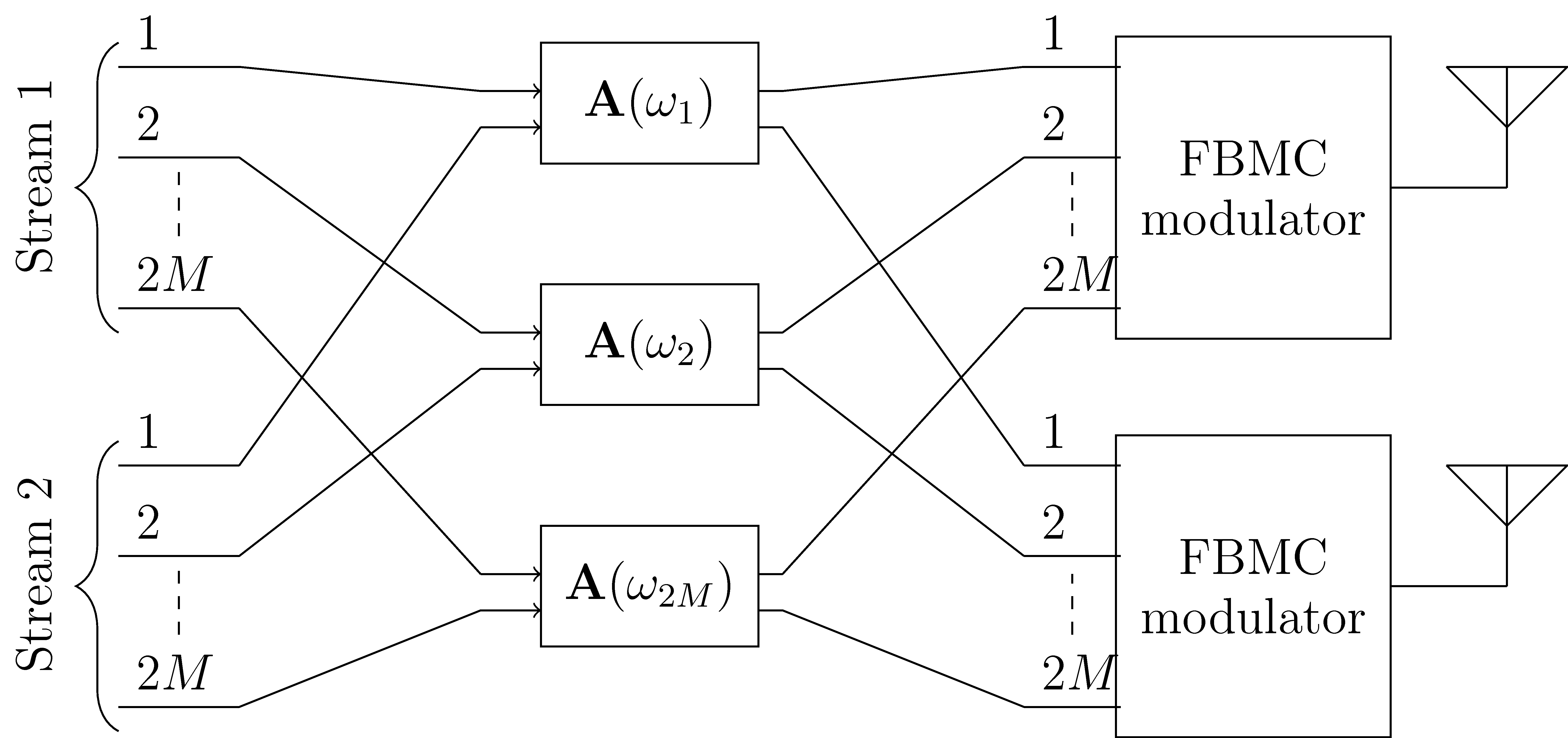}
	\caption{Per-subcarrier precoding at the transmitter.}
	\label{fig:trad_precoder}
\end{figure}

In this paper, in contrast to most of the proposed approaches to deal with channel frequency selectivity, we consider the classical low complexity approach based on one tap per-subcarrier precoding and decoding matrices. However, the assumption of a frequency flat channel at the subcarrier level is not made. A first order approximation of the per-subcarrier MSE for a general MIMO FBMC-OQAM system, {\color{black}including the effects of noise, multi-stream interference (MSI), ISI and ICI}, is proposed, relying not only on the channel frequency response evaluated at this subcarrier but also on its derivatives. This approximation generalizes the one given in \cite{Mestre16} that is valid only for the ZF case.

Furthermore, we optimize the MSE formula to design precoders and decoders in a MU MIMO context. As shown in \fref{fig:syst_model}, we consider a MU MIMO system with one base station (BS) and multiple single-antenna users that are not able to cooperate with each other\footnote{Note that one could straightforwardly apply the results of this paper to a point-to-point (PTP) communication link transmitting with pure spatial multiplexing.}. {\color{black} The users and the BS are assumed to use SDMA \cite[Chap. 10]{paulraj2003introduction}, i.e., they communicate simultaneously using the same time and frequency resources}. Taking into account at the same time the MSI, ISI and ICI caused by channel frequency selectivity during the optimization procedure, we show that even with the very simple chosen structure, one can exploit the degrees of freedom offered by the extra BS antennas to compensate for the distortion due to frequency selectivity. In both the uplink (UL) and downlink (DL) cases, two design criteria are considered, namely zero forcing (ZF) or minimum mean squared error (MMSE). From the asymptotic study at high signal-to-noise ratio (SNR), it is shown that the first order approximation of the distortion can be completely removed as soon as the number of BS antennas is twice as large as the number of users.

\begin{figure}[t!]
	\centering
		\includegraphics[width=0.7\textwidth]{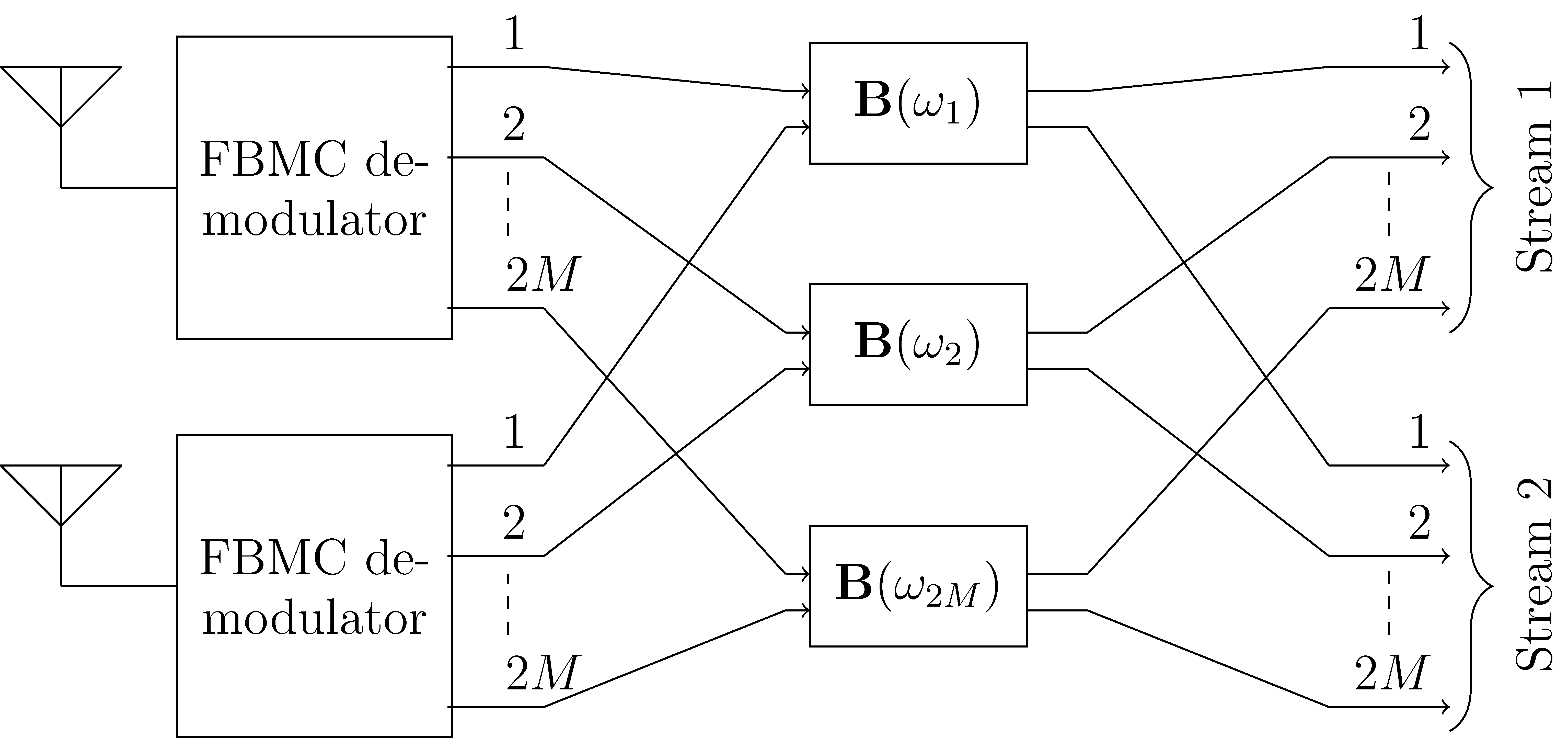}
	\caption{Per-subcarrier decoding at the receiver.}
	\label{fig:trad_receiver}
\end{figure}

The rest of this paper is structured as follows. \sref{section_MSE} details the data model for a general FBMC-OQAM MIMO transceiver and proposes an approximation of the mean squared error (MSE) of the system under strong channel frequency selectivity. \sref{section_optimal_precoder_decoder} optimizes the previously derived MSE formula for a MU MIMO scenario as a function of the linear precoder or decoder and under a ZF or a minimum mean squared error (MMSE) criterion. \sref{section_results} validates the accuracy of MSE approximation and the performance of the linear precoder and decoder through simulations. Finally, \sref{section_conclusion} concludes the paper and appendixes contain the mathematical proof of previous sections.

\subsection{Notations}
Vectors and matrices are denoted by bold lowercase and uppercase letters, respectively. Superscripts $^*$, $^T$ and $^H$ stand for conjugate, transpose and Hermitian operators. $\tr$, $\mathbb{E}$, $\Im$ and $\Re$ denote the trace, expectation, imaginary and real parts respectively. Symbol $O\left(  M^{-\ell}\right)  $ denotes a matrix of possibly increasing dimensions whose entries decay to zero faster than $M^{-\ell}$.

\section{MSE Formulation for General MIMO FBMC-OQAM System Under Strong Channel Frequency Selectivity}
\label{section_MSE} {\color{black}
We will first introduce the system model for a general MIMO FBMC-OQAM transmission and then give an approximation of the MSE at the output of the transceiver chain.

\subsection{General MIMO FBMC-OQAM transmission} \label{subsect:data_model}
Let us consider a MIMO FBMC-OQAM system with $N_T$ and $N_R$ antennas at the transmit and receive sides, respectively. The number of real-valued multicarrier symbols is denoted by $2N_s$ and the number of streams by $S$. 

Multicarrier modulations divide the transmission band into multiple narrow bands. If the number of subcarriers, denoted by $2M$, is large enough with respect to (w.r.t.) the channel delay spread, a common assumption is to assume that the channel is approximately frequency flat inside each sub-band so that precoding (pre-equalization) and decoding (equalization) operations can be performed at the subcarrier level. The block diagrams of the transmitter and receiver are depicted in \fref{fig:trad_precoder} and \fref{fig:trad_receiver}. At the transmitter, the precoding matrix at the $m$-th subcarrier is denoted by $\mat{A}(\omega_m)\in \mathbb{C}^{N_T \times S}$. At the receiver, the decoding matrix at the $m$-th subcarrier is denoted by $\mat{B}(\omega_m)\in \mathbb{C}^{S \times N_R}$.

The real-valued transmitted symbols denoted by $\vect{d}_{l,m} \in \mathbb{R}^{S \times 1}$ are first precoded and then FBMC-OQAM modulated using a prototype pulse $p[n]$ of length $L_p$. The transmitted signal at the different transmit antennas, denoted by $\vect{s}[n]\in \mathbb{C}^{N_T\times 1}$, is given by
\begin{align*}
	\vect{s}[n]&=\sum_{l=0}^{2N_s-1} \sum_{m=0}^{2M-1} \mat{A}(\omega_m)\vect{d}_{l,m} p_{l,m}[n]
\end{align*}
where $p_{l,m}[n]=\frac{j^{l+m}}{M}p[n-lM]e^{j\frac{2\pi}{2M}m(n-\frac{L_p-1}{2})}$. We denote by $\mat{H}(\omega)\in \mathbb{C}^{N_R\times N_T}$ the channel frequency response matrix. The signal at the different receive antennas, denoted by $\vect{r}[n]\in \mathbb{C}^{N_R\times 1}$, is given by
\begin{align*}
	\vect{r}[n]&=\sum_{b=-\infty}^{+\infty} \mathcal{H}[b]\vect{s}[n-b]+\vect{w} [n]
\end{align*}
where $\mathcal{H}[b]=\frac{1}{2\pi}\int_{0}^{2\pi}\mat{H}(\omega)e^{j\omega b}d\omega$ is the channel impulse response which is assumed not to change over the frame time duration. The column vector $\vect{w} [n]$ contains the additive white Gaussian noise samples. The received signal is FBMC-OQAM demodulated using prototype pulse $q[n]$ of length $L_q$. The signal after demodulation and decoding (equalization), at subcarrier $l_0$ and multicarrier symbol $m_0$, denoted by $\vect{x}_{l_0,m_0}\in \mathbb{C}^{S\times 1}$, may be written as
\begin{align*}
\vect{x}_{l_0,m_0}&=\mat{B}(\omega _{m_0})\sum_{n=0}^{L_q-1} \vect{r}[n] q^*_{l_0,m_0}[n]
\end{align*}
where $q_{l_0,m_0}[n]=\frac{j^{l_0+m_0}}{M}\tilde{q}[n-l_0M]e^{j\frac{2\pi}{2M}m_0(n-\frac{L_q-1}{2})}$ and where $\tilde{q}[n]$ is the reversed version of the receive prototype, namely $\tilde{q}[n]=q[L_q-1-n]$. Finally the estimated symbols are obtained by taking the real part, i.e., $\hat{\vect{d}}_{l_0,m_0}=\Re \{\vect{x}_{l_0,m_0}\}$.

\subsection{MSE formulation}
We define the MSE at the output of the transceiver chain corresponding to all streams as\footnote{Observe that we introduce the factor $2$ in order to consider the distortion of the complex symbols, and not the real ones.}
\begin{align}
	\mathrm{MSE}(m)&=2 \mathbb{E} \left(\|\hat{\vect{d}}_{l,m}-{\vect{d}}_{l,m}\|^2\right)\nonumber\\
	&=P_d(m)+N_0 \tr \left[\mat{B}(\omega _{m})\mat{B}(\omega _{m})^H\right]  \label{eq:asymptotic_MSE}
\end{align}
where $N_0$ is the noise power and the expectation is taken over transmitted symbols and noise. Since noise and symbols are uncorrelated, their effect can be separated in the two terms of (\ref{eq:asymptotic_MSE}). The term $P_d(m)$ corresponds to the distortion due to MSI, ISI and ICI. The designs of \fref{fig:trad_precoder} and \fref{fig:trad_receiver} usually rely on channel frequency flatness at the subcarrier level. When the variation of the channel becomes non-negligible, this assumption becomes inaccurate and distortion will increase with the appearance of MSI, ISI and ICI (, i.e., the term $P_d(m)$ increases). To be able to give an analytical expression of $P_d(m)$, we make the following assumptions:
}

$\mathbf{(As1)}$ The actual precoding and decoding matrices implemented at the $m$-th subcarrier result from the evaluation of the functions $\mat{A}(\omega)$ and $\mat{B}(\omega)$ at frequency $\omega_m=\frac{2\pi (m-1)}{2M}$. The precoder, decoder and channel frequency response matrices, $\mat{A}(\omega) \in \mathbb{C}^{N_T \times S}$, $\mat{B}(\omega)\in \mathbb{C}^{S \times N_R}$ and $\mat{H}(\omega)\in \mathbb{C}^{N_R\times N_T}$, are twice differentiable functions of the frequency $\omega$ on the torus $\mathbb{R}%
/2\pi\mathbb{Z}$. 

$\mathbf{(As2)}$ The prototype pulse $p[n]$ is assumed identical at transmit and receive sides, so that $p[n]=q[n]$. It is either symmetric or anti-symmetric in the time domain and it meets the perfect reconstruction\ (PR) conditions. It has length $2M\kappa$, where
$\kappa$ is the overlapping factor. Furthermore, $p[n]$ is obtained by
discretization of a smooth real-valued analog waveform $p(t)$, which is a $\mathcal{C}^{\infty}\left(  \left[  -T_{s}\kappa/2,T_{s}%
\kappa/2\right]  \right)  $ function, so that
\[
p[n]=p\left(  \left(  n-\frac{2M\kappa+1}{2}\right)  \frac{T_{s}}{2M}\right)
,\ n=1,\ldots,2M\kappa
\]
where $T_{s}$ is the multicarrier symbol period. Furthermore, the pulse $p(t)$
and its derivatives are null at the end-points of the support, namely at
$t=\pm T_{s}\kappa/2$.

Thanks to the above assumption, we can define $p^{(r)}[n]$ as the sampled
version of the $r$-th derivative of $p(t)$, that is
\[
p^{(r)}[n]=T_{s}^{r}p^{(r)}\left(  \left(  n-\frac{2M\kappa+1}{2}\right)
\frac{T_{s}}{2M}\right), n=1,\ldots,2M\kappa.
\]
$\mathbf{(As3)}$ The real-valued symbols $\vect{d}_{l,m}$ are independent, identically distributed bounded random variables with zero mean and variance $P_{s}/2$.

We are now in a position to introduce the main result of this section.

\begin{theorem}\label{Theorem}
	Under $\mathbf{(As1)-(As3)}$, the MSE of the complex symbols at the $m$-th subcarrier can
	be expressed as%
	\begin{align}
		P_{d}(m)  &=P_{s}\text{tr}\left[  \left(  \mat{BHA-I}\right)  \left(
		\mat{BHA-I}\right)  ^{H}\right]  \nonumber \\
		&  +\frac{2\eta_{1010}^{(+,-)}}{\left(  2M\right) ^2}\text{tr}\left[  \left(
		\mat{BH}'\mat{A}\right)  \left(  \mat{BH}'\mat{A}%
		\right)  ^{H}\right]  \nonumber\\
		&  +\frac{2\eta_{1010}^{(+,-)}}{\left(  2M\right)  ^2}\Re\text{tr}\left[
		\left(  \mat{BHA-I}\right)  \left(  \mat{BH}''\mat{A}\right)^{H}\right]  \nonumber\\
		&  +\frac{4\left(  \eta_{1010}^{(+,-)}+\eta_{0011}^{(-,+)}\right)  }{\left(  2M\right)  ^2}%
		\text{tr}\left[  \Im\left(  \mat{BHA-I}\right)
		\Im^{T}\left(  \mat{B}\left(  \mat{HA}'\right)
		'\right)  \right]  \nonumber\\
		&  +\frac{4\left(  \eta_{1010}^{(+,-)}+\eta_{0011}^{(-,+)}\right)  }{\left(  2M\right)  ^2}%
		\text{tr}\left[  \Im\left(  \mat{BHA}'\right)
		\Im^{T}\left(  \mat{B}\left(  \mat{HA}\right)
		'\right)  \right] \nonumber\\
		&+O\left(  M^{-2}\right)  \label{eq:asymptotic_distortion},
	\end{align}
	where $\eta_{1010}^{(+,-)}$ and $\eta_{0011}^{(-,+)}$ are pulse-related quantities defined in Appendix \ref{sec:proof_distortion}, $'$ and $''$ refer to the first and second derivatives and where all frequency-depending matrices are evaluated at $\omega=\omega_{m}$, e.g. $\mat{A}=\mat{A}(\omega_m)$, $\mat{H}'=\mat{H}'(\omega_m)$...
\end{theorem}

\begin{proof}
	See Appendix \ref{sec:proof_distortion}.
\end{proof}

\textit{Comments}: equations (\ref{eq:asymptotic_MSE}) and (\ref{eq:asymptotic_distortion}) show that the MSE expression is composed of many terms including the effects of noise, MSI, ISI and ICI. One may recognize some usual terms, i.e., the noise term in (\ref{eq:asymptotic_MSE}) or the first term of (\ref{eq:asymptotic_distortion}) related to the fact that the channel is not perfectly inverted ($\mat{BHA}\neq \mat{I}_{S}$). Those two terms would be the only ones remaining if the channel was frequency flat ($\mat{H}'=\mat{0}$) and the precoder non-frequency selective ($\mat{A}'=\mat{0}$). Those are also the two only terms of distortion in an OFDM system if the cyclic prefix is longer than the channel length and the system well synchronized. Furthermore, the dependence of (\ref{eq:asymptotic_distortion}) on $\mat{H}'$ and $\mat{H}''$ comes directly from the fact that the channel variation breaks the FBMC-OQAM orthogonality while the dependence in the derivatives of the precoder $\mat{A}'$ and $\mat{A}''$ shows that the precoding operations on adjacent subcarriers may influence the MSE at the current subcarrier. Notice that this effect, also known as intrinsic interference, also occurs if the channel is non-varying ($\mat{H}'=\mat{0}$) but the precoder varies over the subcarriers  ($\mat{A}'\neq \mat{0}$). %Interestingly, one will see in the next sections that the two last terms of (\ref{eq:asymptotic_distortion}) depending on the precoder derivatives vanish in the considered MU MIMO situations, giving in the end an expression of the MSE that depends only on the channel conditions ($\mat{H}$, $\mat{H}'$, $\mat{H}''$) and of course on the precoder and decoder at the current subcarrier ($\mat{A}$ and $\mat{B}$), left to optimize.

\section{Linear Precoder and Decoder Design for a MU MIMO system}
\label{section_optimal_precoder_decoder}

%\begin{figure*}[t!]
%	\centering
%	\def\arraystretch{1.3}
%	\begin{tabular}{ |c | c | c| }
%		\hline & Decoder (Uplink, MAC channel) & Precoder (Downlink, BC channel) \\ 
%		& $\mat{A}_\mathrm{UL}=\xi_\mathrm{UL} \mat{I}_{N_U}$, $\mat{H}_\mathrm{UL} \in \mathbb{C}^{N\times N_U}$ & $\mat{B}_\mathrm{DL}=\xi_\mathrm{DL} \mat{I}_{N_U},\mat{H}_\mathrm{DL} \in \mathbb{C}^{N_U\times N}$ \\\hline
%		ZF ($\mat{B}\mat{H}\mat{A}=\mat{I}_{N_U}$) & $(\xi_\mathrm{UL}^{\mathrm{ZF}}(\omega))'=0$ & $(\xi_\mathrm{DL}^{\mathrm{ZF}}(\omega))'\neq 0$ \\ 		\hline  
%		MMSE ($\mat{B}\mat{H}\mat{A}\neq\mat{I}_{N_U}$)& $(\xi_\mathrm{UL}^{\mathrm{MMSE}}(\omega))'=0$ & $(\xi_\mathrm{DL}^{\mathrm{MMSE}}(\omega))'\neq 0$ \\
%		\hline  
%	\end{tabular}
%	\caption{Summary of the different designs under consideration and their respective assumptions.}
%	\label{summary_designs}
%\end{figure*}
\begin {table*}[t!]
\caption {Summary of the different designs under consideration and their respective assumptions.} \label{summary_designs} 
\centering
\def\arraystretch{1.3}
\begin{tabular}{ |c | c | c| }
	\hline & Decoder (Uplink, MAC channel) & Precoder (Downlink, BC channel) \\ 
	& $\mat{A}_\mathrm{UL}=\xi_\mathrm{UL} \mat{I}_{N_U}$, $\mat{H}_\mathrm{UL} \in \mathbb{C}^{N\times N_U}$ & $\mat{B}_\mathrm{DL}=\xi_\mathrm{DL} \mat{I}_{N_U},\mat{H}_\mathrm{DL} \in \mathbb{C}^{N_U\times N}$ \\\hline
	ZF ($\mat{B}\mat{H}\mat{A}=\mat{I}_{N_U}$) & $\xi_\mathrm{UL}^{\mathrm{ZF}}$ independent of $\omega$ & $\xi_\mathrm{DL}^{\mathrm{ZF}}$ depends on $\omega$\\ 		\hline  
	MMSE ($\mat{B}\mat{H}\mat{A}\neq\mat{I}_{N_U}$)& $\xi_\mathrm{UL}^{\mathrm{MMSE}}$ independent of $\omega$ & $\xi_\mathrm{DL}^{\mathrm{MMSE}}$ depends on $\omega$ \\
	\hline  
\end{tabular}
\end {table*}

The goal of this section is to optimize the general MSE formulation of (\ref{eq:asymptotic_MSE}), applied to a MU MIMO scenario, as a function of the linear precoding or decoding matrices. As shown in \fref{fig:syst_model}, we consider a MU MIMO system with one base station (BS) equipped with $N$ antennas and $N_U$ users, each one equipped with a single antenna and not able to cooperate with each other\footnote{{\color{black}Note that the approach could be generalized to the case where each user terminal is equipped with multiple antennas, although the extension does not seem trivial.}}. {\color{black} The users and the BS are assumed to use SDMA, i.e., they communicate simultaneously using the same time and frequency resources}. The number of streams is equal to $S=N_U$ with $N\geq N_U$. The channel frequency response matrix $\mat{H}(\omega)$ is assumed to be perfectly known by the BS. For the sake of clarity, $\mat{H}(\omega)$ is denoted by $\mat{H}_{\mathrm{DL}}(\omega) \in \mathbb{C}^{N_U \times N}$ when referred to the specific DL scenario and $\mat{H}_{\mathrm{UL}}(\omega) \in \mathbb{C}^{N \times N_U}$ resp. in the UL case. 

While optimizing the decoder in UL (resp. precoder in DL), the precoder (resp. decoder) at the other end is fixed to a real positive power normalization scalar $\xi(\omega)$ since the users cannot collaborate. In the following, a per-subcarrier total transmit power $P_T$ constraint is considered. 
\begin{align}
\label{constraint_power}
\tr \left[\mat{A}\mat{A}^H\right]&=P_T.
\end{align}
%This constraint is logical for the UL since the users have no knowledge of the channel and cannot perform specific power allocation. Regarding the DL case, the BS could allocate different powers to different subcarrier. However, this paper addresses low complexity designs which can be performed at the subcarrier level and then avoids joint optimization of the power over all subcarriers. 
Two design criteria will be investigated, namely the ZF criterion and the MMSE criterion. {\color{black}A summary of the different designs under study with their corresponding assumptions is given in Table~\ref{summary_designs}. Note that in UL, the users cannot not precode the streams, so that $\xi_\mathrm{UL}(\omega)$ is frequency non-selective. Conversely, in the DL, the BS pre-equalizes the channel at the subcarrier level. This processing depends on the channel frequency response at this subcarrier and hence, the normalization factor $\xi_\mathrm{DL}(\omega)$ will generally depend on frequency.} Finally, the computation complexity of the proposed designs will be studied.

\subsection{Zero Forcing Design}
For this design, a channel inverting constraint is considered, namely
\begin{align}
	\label{constraint_inversion}
	\mat{B}\mat{H}\mat{A}&=\mat{I}_{N_U}.
\end{align}
The channel matrix $\mat{H}$ is assumed full rank, which is a quite natural assumption in the considered MU MIMO scenario. Using (\ref{constraint_inversion})  and the fact that $(\mat{B}\mat{H}\mat{A})'=\mat{B}(\mat{HA})'+\mat{B}'\mat{HA}=\mat{0}$, many terms of the distortion expression of \eref{eq:asymptotic_distortion} vanish and the MSE in (\ref{eq:asymptotic_MSE}) simplifies to
\begin{align}
\label{MSE_ZF}
\MSE(m)&=\alpha\tr \left[\left(\mat{B}\mat{H}'\mat{A}\right)\left(\mat{B}\mat{H}'\mat{A}\right)^H\right] \nonumber\\
&-(2 \alpha +2\beta) \tr \left[\Im(\mat{B}\mat{H}\mat{A}')\Im(\mat{B}'\mat{H}\mat{A})^T\right]\nonumber\\
&+N_0 \tr \left[\mat{B}\mat{B}^H\right] + O\left(2M^{-2}\right).
\end{align}
where $\alpha=\frac{2\eta_{1010}^{(+,-)}}{\left(  2M\right)  ^{2}}$, $\beta=\frac{2\eta_{0011}^{(-,+)}}{\left(  2M\right)  ^{2}}$.

\subsubsection{Linear Decoder (Multiple Access Channel (MAC), Uplink)}
In the UL case, $\mat{H}_{\mathrm{UL}}$ and $\mat{H}_{\mathrm{UL}}' \in \mathbb{C}^{N\times N_U}$ correspond to the tall channel frequency response matrix and its derivative evaluated at the subcarrier of interest. From the power normalization \eref{constraint_power} and channel inversion \eref{constraint_inversion} constraints, the general solution of the problem can be written in the following form
\begin{align}
	\mat{A}^{\mathrm{ZF}}_{\mathrm{UL}}&=\xi^{\mathrm{ZF}}_{\mathrm{UL}}\mat{I}_{N_U}\nonumber \\
	\mat{B}^{\mathrm{ZF}}_{\mathrm{UL}}&=\frac{1}{\xi^{\mathrm{ZF}}_{\mathrm{UL}}}\left( \mat{H}^\dagger + \tilde{\mat{B}}\mat{P}_{\mathrm{UL}} \right) \nonumber
\end{align}
where $\mat{H}^\dagger_{\mathrm{UL}}=(\mat{H}_{\mathrm{UL}}^H\mat{H}_{\mathrm{UL}})^{-1}\mat{H}_{\mathrm{UL}}^H$, $\mat{P}_{\mathrm{UL}}=\mat{I}_{N}-\mat{H}_{\mathrm{UL}}\mat{H}^\dagger_{\mathrm{UL}}$, $\xi^{\mathrm{ZF}}_{\mathrm{UL}}={\sqrt{P_T/N_U}}$ and where $\tilde{\mat{B}}$ is a $N_U \times N$ matrix to be optimized. This shows that the decoder can be written as the left pseudo-inverse of the channel plus a matrix lying on the left null space of $\mat{H}_{\mathrm{UL}}$. In the trivial case $N=N_U$, the decoder is the inverse of the channel since there are no extra degrees of freedom. One can check that the second term of the distortion in \eref{MSE_ZF} is null due to the fact that $\Im(\mat{B}^{\mathrm{ZF}}_{\mathrm{UL}}\mat{H}_{\mathrm{UL}}(\mat{A}^{\mathrm{ZF}}_{\mathrm{UL}})')=\Im\left( {(\xi^{\mathrm{ZF}}_{\mathrm{UL}})'}/{\xi^{\mathrm{ZF}}_{\mathrm{UL}}}\right) \mat{I}_{N_U}=\vect{0}$ with $\xi^{\mathrm{ZF}}_{\mathrm{UL}}$ purely real and frequency non-selective. Therefore, the optimization problem can be turned into the minimization of a quadratic expression in $\tilde{\mat{B}}$
\begin{align} \label{min_prob_decoder}
	\min_{\tilde{\mat{B}}}& \ \alpha \ \tr\left[\left(\mat{H}^\dagger_{\mathrm{UL}}\mat{H}_{\mathrm{UL}}'+\tilde{\mat{B}}\mat{P}_{\mathrm{UL}}\mat{H}_{\mathrm{UL}}'\right)\left(\mat{H}^\dagger_{\mathrm{UL}}\mat{H}_{\mathrm{UL}}'+\tilde{\mat{B}}\mat{P}_{\mathrm{UL}}\mat{H}_{\mathrm{UL}}'\right)^H\right] +\frac{N_0 N_U}{P_T} \tr \left[\left(\mat{H}_{\mathrm{UL}}^H\mat{H}_{\mathrm{UL}}\right)^{-1} +\tilde{\mat{B}}\mat{P}_{\mathrm{UL}} \tilde{\mat{B}}^H\right].
\end{align}
Setting the derivative of this expression with respect to $\tilde{\mat{B}}^*$ to $\mat{0}$, we find that the optimum solution is such that
\begin{align} \label{null_space_projection}
	\tilde{\mat{B}}%&=-\mat{H}^\dagger_{\mathrm{UL}}\mat{H}_{\mathrm{UL}}'(\mat{P}_{\mathrm{UL}}\mat{H}_{\mathrm{UL}}')^H\left(\mat{H}_{\mathrm{UL}}'(\mat{P}_{\mathrm{UL}}\mat{H}_{\mathrm{UL}}')^H+\frac{N_0N_U}{\alpha}\mat{I}_{N}\right)^{-1}\\
	&=-\mat{H}^\dagger_{\mathrm{UL}}\mat{H}_{\mathrm{UL}}'\left(\mat{H}_{\mathrm{UL}}'^H\mat{P}_{\mathrm{UL}}\mat{H}_{\mathrm{UL}}'+\frac{N_0N_U}{P_T\alpha}\mat{I}_{N_U}\right)^{-1} \mat{H}_{\mathrm{UL}}'^H 
\end{align}
where we used the matrix inversion lemma.

\subsubsection{Linear Precoder (Broadcast Channel (BC), Downlink)}
In the DL case, $\HDL, \HDL' \in \mathbb{C}^{N_U\times N}$ denote the fat channel frequency response matrix and its derivative evaluated at the subcarrier of interest. From the constraints \eref{constraint_power} and \eref{constraint_inversion}, the general solution can be written as
\begin{align*}
	\mat{A}^{\mathrm{ZF}}_{\mathrm{DL}}&=\frac{1}{\xi^{\mathrm{ZF}}_{\mathrm{DL}}} \left( \mat{H}^\dagger_{\mathrm{DL}} + \mat{P}_{\mathrm{DL}}\tilde{\mat{A}} \right) \nonumber\\
	\mat{B}^{\mathrm{ZF}}_{\mathrm{DL}}&=\xi^{\mathrm{ZF}}_{\mathrm{DL}}\mat{I}_{N_U}
\end{align*}
where $\mat{H}^\dagger_{\mathrm{DL}}=\HDL^H(\HDL\HDL^H)^{-1}$, $\mat{P}_{\mathrm{DL}}=\mat{I}_{N}-\mat{H}^\dagger_{\mathrm{DL}}\HDL$ and $\xi^{\mathrm{ZF}}_{\mathrm{DL}}=\sqrt{\tr\left((\HDL\HDL^H)^{-1}+\tilde{\mat{A}}^H\mat{P}_{\mathrm{DL}}\tilde{\mat{A}}\right)/P_T}$. As in the decoder case, the second term of the distortion in \eref{MSE_ZF} also disappears due to the fact that $\Im((\mat{B}^{\mathrm{ZF}}_{\mathrm{DL}})'\HDL\mat{A}^{\mathrm{ZF}}_{\mathrm{DL}})=\Im({(\xi^{\mathrm{ZF}}_{\mathrm{DL}})'}/{\xi^{\mathrm{ZF}}_{\mathrm{DL}}})\mat{I}_{N_U}=\vect{0}$ with $\xi$ purely real. The optimization problem then simplifies to
\begin{align*}
	\min_{\tilde{\mat{A}}} & \ \alpha \ \tr \left[\left(\HDL'\mat{H}^\dagger_{\mathrm{DL}} +  \HDL'\mat{P}_{\mathrm{DL}} \tilde{\mat{A}}  \right)\left(\HDL'\mat{H}^\dagger_{\mathrm{DL}} + \HDL'\mat{P}_{\mathrm{DL}} \tilde{\mat{A}} \right)^H\right]+\frac{N_0 N_U}{P_T}\tr\left[(\HDL\HDL^H)^{-1}+\tilde{\mat{A}}^H\mat{P}_{\mathrm{DL}}\tilde{\mat{A}}\right],
\end{align*}
the solution of which is, after applying matrix inversion lemma,
\begin{align*}
	\tilde{\mat{A}}&=- \HDL^{'H} \left( \HDL'\mat{P}_{\mathrm{DL}} \HDL^{'H}+ \frac{N_0 N_U}{P_T\alpha} \mat{I}_{N_U} \right)^{-1} \HDL'\mat{H}^\dagger_{\mathrm{DL}}. 
\end{align*}
One can check that the asymptotic MSE of the optimized precoder and decoder will be exactly the same if the channels are the Hermitian of one another, \emph{i.e.} $\HDL=\mat{H}_{\mathrm{UL}}^H$.

\subsubsection{Asymptotic Study at Low and High SNR}
\label{asymptotic_anal}
We concentrate here on the behavior of the optimized linear decoder in the UL (MAC) channel. Similar conclusions also hold for the precoder in the DL (BC) channel. We assume that the number of users $N_U$ and the transmit power $P_T$ remain constant while we let $N_0$ go to $0$ or $+\infty$ (high and low SNR respectively). At low SNR ($N_0 \rightarrow +\infty$), the expression in \eref{null_space_projection} tends to zero ($\tilde{\mat{B}}\rightarrow \mat{0}$) and the optimized decoder converges to 
\begin{align*}
	\lim_{N_0\rightarrow \infty} \mat{B}^{\mathrm{ZF}}_{\mathrm{UL}}&= \frac{1}{\xi^{\mathrm{ZF}}_{\mathrm{UL}}}\mat{H}^\dagger_{\mathrm{UL}}.
\end{align*}
As one would expect, when noise power is large, the distortion caused by channel selectivity is comparatively negligible. The best thing to do is to use the classical pseudo-inverse of the channel to combine the signals of each antenna.

At high SNR, assuming that $\mat{H}_{\mathrm{UL}}'$ is of full rank $N_U$, the decoder converges to a limit that depends on the rank of $\mat{P}_{\mathrm{UL}}$. {\color{black}Indeed,  two cases must be considered depending on whether matrix $\mat{H}_{\mathrm{UL}}'^H\mat{P}_{\mathrm{UL}}\mat{H}_{\mathrm{UL}}'$ is invertible or not.} One can rewrite $\mat{P}_{\mathrm{UL}}$ as a function of the singular value decomposition (SVD) of $\mat{H}_{\mathrm{UL}}$
\begin{align*}
\mat{H}_{\mathrm{UL}}=\begin{bmatrix}
\mat{U}_1 & \mat{U}_2
\end{bmatrix}\begin{bmatrix}
\mat{\Sigma}_{N_U\times N_U}& \mat{0}_{N-N_U\times N_U}^H
\end{bmatrix}^H \mat{V}^H.
\end{align*}
We then find $\mat{P}_{\mathrm{UL}}=\mat{U}_2\mat{U}_2^H$ where $\mat{U}_2$ is the $N\times N-N_U$ matrix composed of the $N-N_U$ left singular vectors of $\mat{H}_{\mathrm{UL}}$ associated to its zero singular values. It is then straightforward to see that the rank of $\mat{P}_{\mathrm{UL}}$ is the dimension of the left null space of $\mat{H}_{\mathrm{UL}}$,\textit{i.e.} $N-N_U$. First, if $N-N_U\geq N_U$, matrix $\mat{H}_{\mathrm{UL}}'^H\mat{P}_{\mathrm{UL}}\mat{H}_{\mathrm{UL}}'$ is full rank and the limit becomes
\begin{align*}
	\lim_{N_0\rightarrow 0}& \tilde{\mat{B}}= -\mat{H}^\dagger_{\mathrm{UL}}\mat{H}_{\mathrm{UL}}'\left(\mat{H}_{\mathrm{UL}}'^H\mat{P}_{\mathrm{UL}}\mat{H}_{\mathrm{UL}}'\right)^{-1} \mat{H}_{\mathrm{UL}}'^H.
\end{align*}
Replacing this expression of $\tilde{\mat{B}}$ into \eref{min_prob_decoder}, it can be seen that the limit of the asymptotic MSE at high SNR will tend to zero. This means that for twice as many antennas as the number of served users, we can completely remove the first order approximation of the distortion caused by channel frequency selectivity.

As for the case $N-N_U < N_U$, using the fact that $\mat{P}_{\mathrm{UL}}=\mat{U}_2\mat{U}_2^H$, one can reapply the matrix inversion lemma on $\tilde{\mat{B}}\mat{P}_{\mathrm{UL}}$ in order to show that the limit becomes
\begin{align*}
	&\lim_{N_0\rightarrow 0} \tilde{\mat{B}}\mat{P}_{\mathrm{UL}}= - \mat{H}^\dagger_{\mathrm{UL}}\mat{H}_{\mathrm{UL}}'\mat{H}_{\mathrm{UL}}'^{H}\mat{U}_2\left(\mat{U}_2^H\mat{H}_{\mathrm{UL}}'\mat{H}_{\mathrm{UL}}'^{H}\mat{U}_2   \right)^{-1}\mat{U}_2^H.
\end{align*}
In this case, the noise term of the MSE will tend to zero but the first order approximation of the distortion will only be partially compensated for.

We can conclude that the optimized ZF decoder and precoder can be written in a compact expression as the pseudo-inverse of the channel plus a matrix lying on the null space of the channel. This design can compensate for the degradation due to channel frequency selectivity and even completely remove the first order approximation of the distortion for twice as many BS antennas as the number of served users.

%\begin{figure*}[b!]
%	\centering
%	\def\arraystretch{1.3}
%	\begin{tabular}{ |c | c | c| }
%		\hline & Decoder (Uplink, MAC channel) & Precoder (Downlink, BC channel) \\\hline
%		Classical ZF  & $O(2NN_U^2+N_U^3)$ & $O(3NN_U^2+N_U^3)$ \\ 		\hline  
%		Opt. ZF  & $O(4NN_U^2+4N^2N_U+2N_U^3)$ & $O(5NN_U^2+4N^2N_U+2N_U^3)$ \\ 		\hline  
%		Classical MMSE  & $O(2N^2N_U+N^3)$ & $O(2N^2N_U+NN_U^2+N^3)$  \\ 		\hline  
%		Opt. MMSE & $O(5N^2N_U+N^3)$& $O(6N^2N_U+3NN_U^2+2N_U^3+N^3)$ \\
%		\hline  
%	\end{tabular}
%	\caption{Complexity of calculating the proposed precoders and decoders.}
%	\label{complexity_design}
%\end{figure*}

\begin {table*}[t!]
\caption {Complexity of calculating the proposed precoders and decoders.} \label{complexity_design} 
\centering
\def\arraystretch{1.3}
\begin{tabular}{ |c | c | c| }
	\hline & Decoder (Uplink, MAC channel) & Precoder (Downlink, BC channel) \\\hline
	Classical ZF  & $O(2NN_U^2+N_U^3)$ & $O(3NN_U^2+N_U^3)$ \\ 		\hline  
	Opt. ZF  & $O(4NN_U^2+4N^2N_U+2N_U^3)$ & $O(5NN_U^2+4N^2N_U+2N_U^3)$ \\ 		\hline  
	Classical MMSE  & $O(2N^2N_U+N^3)$ & $O(2N^2N_U+NN_U^2+N^3)$  \\ 		\hline  
	Opt. MMSE & $O(5N^2N_U+N^3)$& $O(6N^2N_U+3NN_U^2+2N_U^3+N^3)$ \\
	\hline  
\end{tabular}
\end {table*}

\subsection{Minimum Mean Squared Error Design}
The previous designs rely on a ZF criterion which restricts the solution domain. In the following, we do not make this assumption and we look at the general MMSE design which will achieve an optimized performance. Indeed, for low SNR situations or highly selective subchannels, inverting the channel might strongly degrade the performance. Furthermore, the channel matrix $\mat{H}$ does not generally need to be full rank.

%Different challenges remains concerning the MMSE precoder as explained below. {\color{black}During our study, we saw that in this case, the MSE expression is not restricted to be positive. That might be a problem at high SNR. However through simulations, what happens is that it only occurs beyond SNR of 40dB and for very highly frequency selective channels, such as the Veh B channel model. Then I would propose not too add this term in the distortion formula...}
\subsubsection{Linear Decoder (Multiple Access Channel, Uplink)}
For the decoder case, due to the power normalization constraint \eref{constraint_power}, we impose 
\begin{align*}
	\mat{A}^{\mathrm{MMSE}}_{\mathrm{UL}}&=\xi^{\mathrm{MMSE}}_{\mathrm{UL}}\mat{I}_{N_U}\\
	\mat{B}^{\mathrm{MMSE}}_{\mathrm{UL}}&=\frac{1}{\xi^{\mathrm{MMSE}}_{\mathrm{UL}}}\hat{\mat{B}}
\end{align*}where $\xi^{\mathrm{MMSE}}_{\mathrm{UL}}=\sqrt{\frac{P_T}{N_U}}$ and $\hat{\mat{B}}=\xi^{\mathrm{MMSE}}_{\mathrm{UL}}\mat{B}^{\mathrm{MMSE}}_{\mathrm{UL}}$ is defined to clarify the following expressions by suppressing the dependence in $\xi^{\mathrm{MMSE}}_{\mathrm{UL}}$. Hence, the imaginary terms of the distortion in \eref{eq:asymptotic_distortion} again disappear due to $(\mat{A}^{\mathrm{MMSE}}_{\mathrm{UL}})'=\mat{0}$ (the precoder is frequency independent) and the optimization problem takes the following quadratic form in $\hat{\mat{B}}$
\begin{align*}
\min_{\hat{\mat{B}}} \MSE(m)&=\tr \left[(\hat{\mat{B}}\mat{H}_\mathrm{UL}-\mat{I})(\hat{\mat{B}}\mat{H}_\mathrm{UL}-\mat{I})^H\right]\nonumber\\
&+\alpha\tr \left[\left(\hat{\mat{B}}\mat{H}_\mathrm{UL}'\right)\left(\hat{\mat{B}}\mat{H}_\mathrm{UL}'\right)^H\right]\nonumber\\
&+\alpha \Re \tr \left[(\hat{\mat{B}}\mat{H}_\mathrm{UL}-\mat{I}) (\hat{\mat{B}}\mat{H}_\mathrm{UL}'')^H\right] \nonumber\\
&+\frac{N_0N_U}{P_T}\tr \left[\hat{\mat{B}}\hat{\mat{B}}^H\right].
\end{align*}
Setting the derivative of this expression with respect to $\hat{\mat{B}}^*$ to $\mat{0}$ yields the MMSE decoder given by,
\begin{align} 
	&\hat{\mat{B}}=\left(\mat{H}_{\mathrm{UL}}^H+\frac{\alpha}{2}\mat{H}_{\mathrm{UL}}''^H\right)\left(\mat{H}_{\mathrm{UL}}\mat{H}_{\mathrm{UL}}^H+\alpha\mat{H}_{\mathrm{UL}}'\mat{H}_{\mathrm{UL}}'^H+\frac{\alpha}{2}\left( \mat{H}_{\mathrm{UL}}\mat{H}_{\mathrm{UL}}''^H+\mat{H}_{\mathrm{UL}}''\mat{H}_{\mathrm{UL}}^H\right)+\frac{N_0N_U}{P_T}\mat{I}_N  \right)^{-1} 	\label{MMSE_decoder}
\end{align}

\subsubsection{Linear Precoder (Broadcast Channel, Downlink)} In the precoder case, due to the normalization constraint, we have 
\begin{align*}
\mat{A}^{\mathrm{MMSE}}_{\mathrm{DL}}&=\frac{1}{\xi^{\mathrm{MMSE}}_{\mathrm{DL}}}\hat{\mat{A}}\nonumber\\
\mat{B}^{\mathrm{MMSE}}_{\mathrm{DL}}&=\xi^{\mathrm{MMSE}}_{\mathrm{DL}}\mat{I}_{N_U}
\end{align*}
with $\xi^{\mathrm{MMSE}}_{\mathrm{DL}}=\sqrt{\frac{\tr\left[\hat{\mat{A}}\hat{\mat{A}}^H \right]}{P_T}}$. As opposed to all of the previous designs, the imaginary terms of the distortion in (\ref{eq:asymptotic_distortion}) do not cancel out. The optimization of those two terms is difficult due to the dependence in $(\mat{A}^{\mathrm{MMSE}}_{\mathrm{DL}})'$. The derivative implies that the optimization of the precoder of one subcarrier depends on the neighboring subcarriers and the optimization can no longer be done locally at the subcarrier level, which increases the problem complexity and is not comparable to the other designs. Hence, we propose to impose an additional constraint which cancels the imaginary terms of (\ref{eq:asymptotic_distortion}), i.e., $\Im(\HDL\mat{A}^{\mathrm{MMSE}}_{\mathrm{DL}})=\mat{0}$. This somehow means that we have a ZF design on the imaginary part of $\HDL\mat{A}^{\mathrm{MMSE}}_{\mathrm{DL}}$ and a MMSE design on its real part. We then have to minimize the following Lagrangian formulation including the constraint (via the Lagrange multiplier $\mat{\Psi}$)
\begin{align*}
\min_{\hat{\mat{A}}}L&=\tr \left[(\HDL\hat{\mat{A}}-\mat{I})(\HDL\hat{\mat{A}}-\mat{I})^H\right]\nonumber\\
&+\alpha\tr \left[\left(\HDL'\hat{\mat{A}}\right)\left(\HDL'\hat{\mat{A}}\right)^H\right]\nonumber\\
&+\alpha \Re \tr \left[(\HDL\hat{\mat{A}}-\mat{I}) (\HDL^{''}\hat{\mat{A}})^H\right]\nonumber\\
&+\frac{N_0 N_U}{P_T} \tr\left[\hat{\mat{A}}\hat{\mat{A}}^H\right]\nonumber\\
&+j\tr \left[\mat{\Psi}^T(\HDL\hat{\mat{A}}-\HDL^*\hat{\mat{A}}^*)\right].
\end{align*}
Setting the derivative of $L$ with respect to $\hat{\mat{A}}$ to $\mat{0}$ yields
\begin{align}
\hat{\mat{A}}&= \left(\HDL^H\HDL +\alpha \HDL'^H\HDL'+\frac{\alpha}{2} \left(\HDL^H \HDL''+\HDL''^H\HDL\right)+\frac{N_0 N_U}{P_T}\mat{I}_{N}\right)^{-1} \left(j\HDL^H\mat{\Psi}+\HDL^H+\frac{\alpha}{2}\HDL''^H\right), 	\label{MMSE_precoder}
\end{align}
where the value of $\mat{\Psi}$ is fixed thanks to the constraint $\Im(\HDL\mat{A})=\mat{0}$. Denoting $\mat{X}=\HDL^H\HDL +\frac{\alpha}{2} \left[\HDL^H \HDL''+\HDL''^H\HDL\right]+\alpha \HDL'^H\HDL'+\frac{N_0 N_U}{P_T}\mat{I}_{N}$, we find
\begin{align*}
\mat{\Psi}&=-\left( \Re\left(\HDL\mat{X}^{-1}\HDL^H\right)\right) ^{-1}\Im(\HDL\mat{X}^{-1} (\HDL^H+\frac{\alpha}{2}\HDL''^H ) ).
\end{align*}
{\color{black}\subsection{Complexity of computation of the proposed precoders and decoders}

Table~\ref{complexity_design} gives an order of complexity of computing the proposed optimized designs with respect to classical designs. By classical designs, we mean precoders and decoders that rely on the hypothesis of channel frequency flatness at the subcarrier level, i.e., $\mat{H}'(\omega_m)=\mat{H}''(\omega_m)=\mat{0}$. For the calculation, only matrix multiplications and inversions are taken into account given that they are the most complex operations. It is assumed that for general matrices $\mat{D}\in \mathbb{C}^{l \times m},\mat{E} \in \mathbb{C}^{m \times n},\mat{F} \in \mathbb{C}^{m \times m}$, performing matrix multiplication $\mat{DE}$ has complexity $O(lmn)$ and matrix inversion $\mat{F}^{-1}$ has complexity $O(m^3)$. One can check that the calculation complexity of the optimized designs remains similar to the classical. Note that the designs in DL are more complex since they require one more matrix multiplication for the calculation of $\xi_\mathrm{DL}$. Furthermore, the opt. MMSE precoder is slightly more complex than the opt. MMSE decoder due to the required calculation of $\mat{\Psi}$.}

\section{Simulation Results}
\label{section_results}

\begin{figure}[b!]
	\centering
	\includegraphics[width=4.5in]{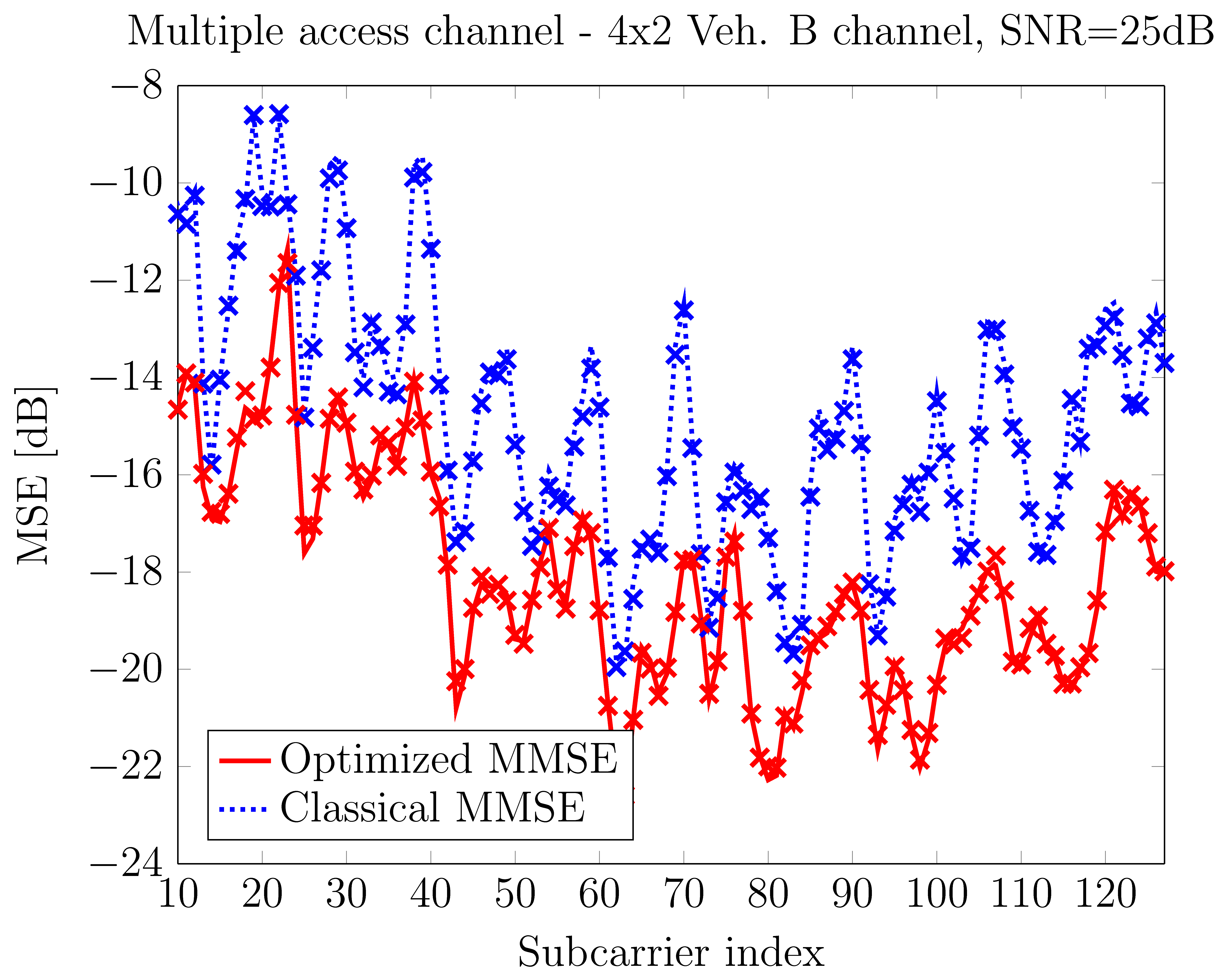}
	%	\resizebox{0.4\textwidth}{!}{%
	%		\Large
	%		\input{fig/Opt_decoder_MSE.tex}
	%	}
	\caption{The optimized MMSE decoder clearly outperforms the classical MMSE decoder. The asymptotic approximation of the MSE represented in a solid line matches perfectly the simulated MSE in crosses.}
	\label{opt_decoder_MSE}
\end{figure}

\begin{figure}[t!]
	\centering
	%	\input{fig/Opt_decoder_SER_VehA.tex}
%	\resizebox{0.45\textwidth}{!}{%
%		\Large
%		\input{fig/Opt_decoder_SER_VehA.tex}
%	}
			\includegraphics[width=4.5in]{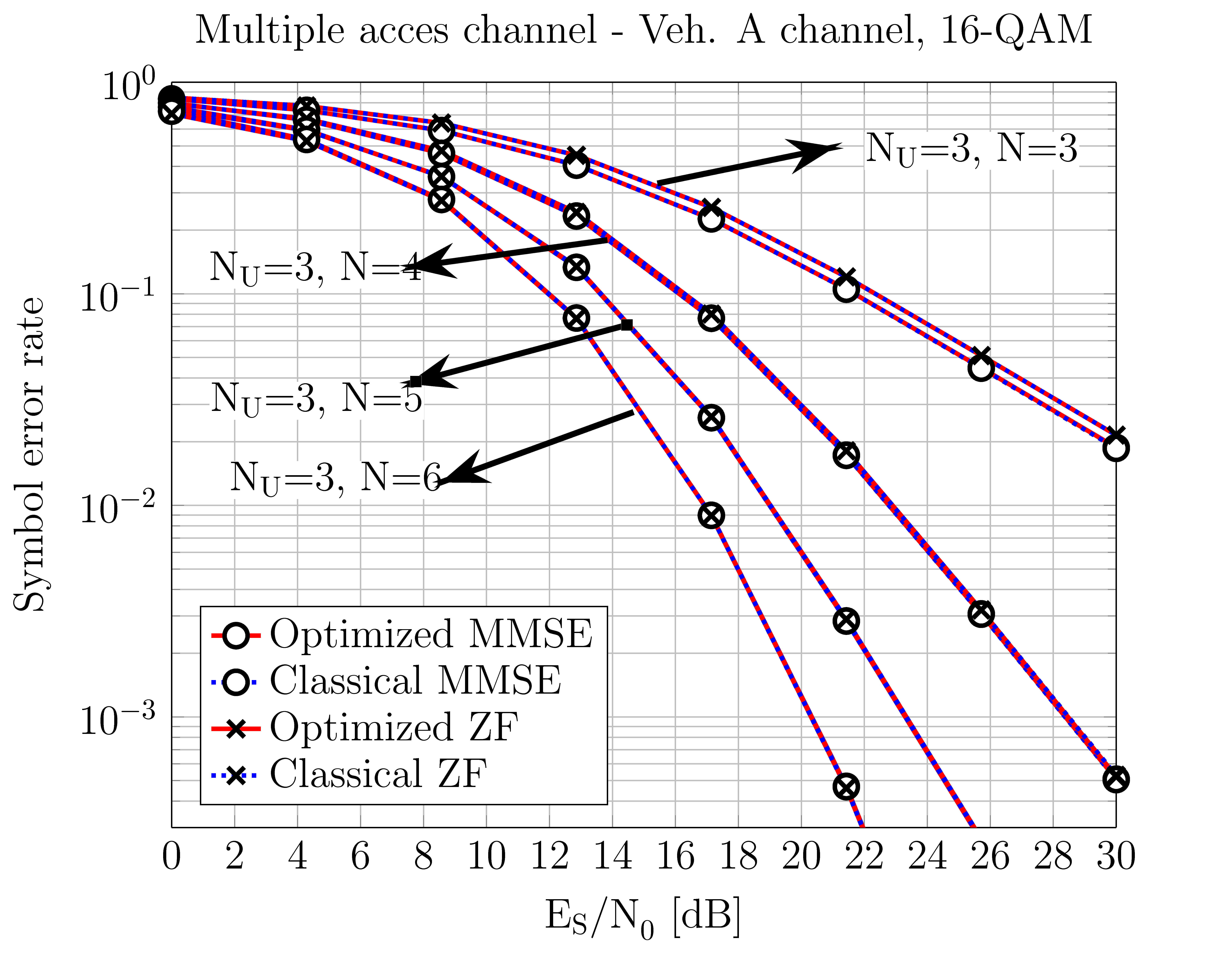}
	\vspace{-1em}
	\caption{SER of the optimized and classical ZF/MMSE decoders and Veh. A channel model.}
	\label{opt_decoder_SER_VehA}
\end{figure}

\begin{figure}[!t]
	\centering
%	\resizebox{0.45\textwidth}{!}{%
%		\Large
%		\input{fig/decoder_MMSE_ZF_16QAM_VehB2}
%	}
		\includegraphics[width=4.5in]{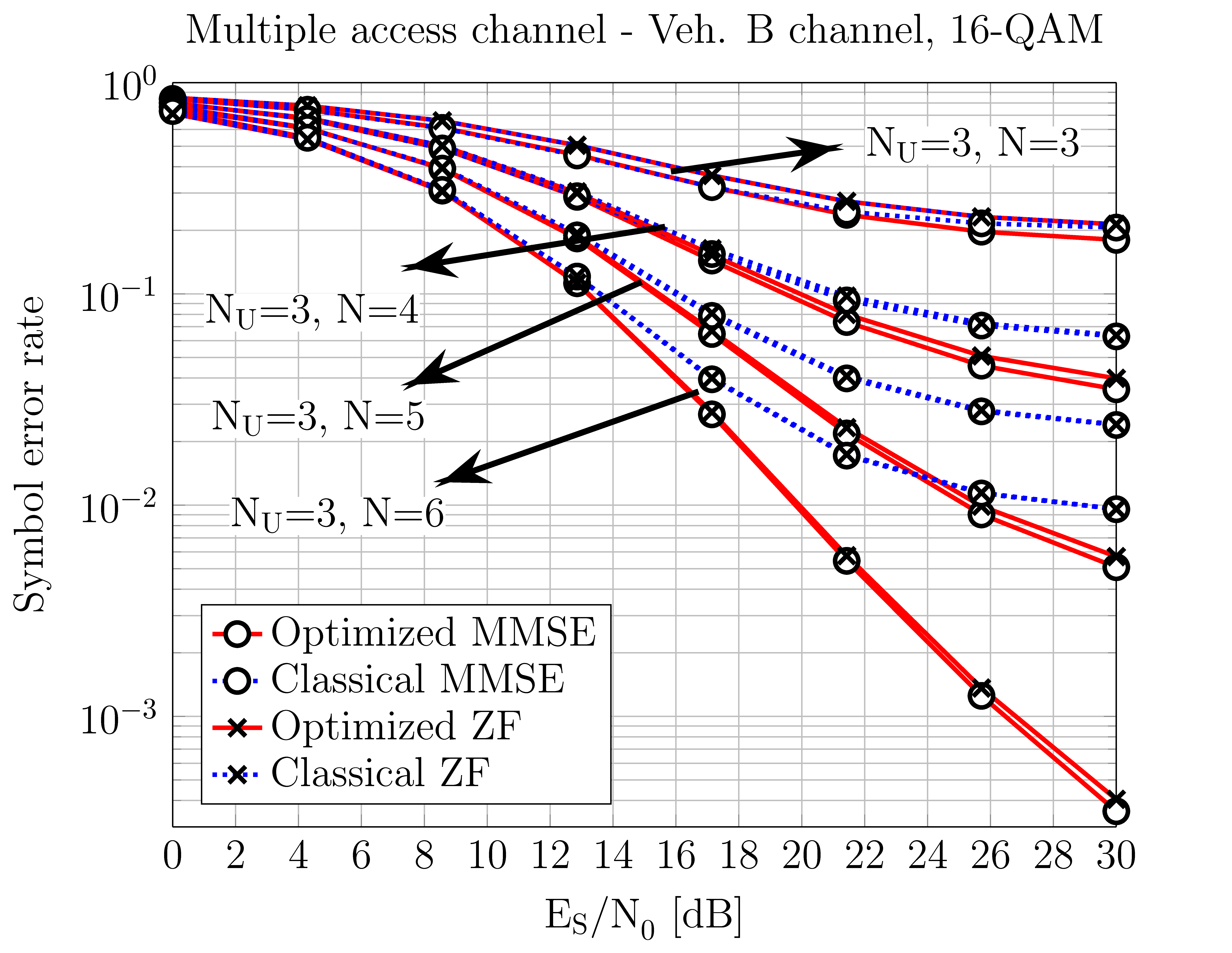}
	\vspace{-1em}
	\caption{SER of the optimized and classical ZF/MMSE decoder and Veh. B
		channel model}
	\label{fig:decoder_MMSE_ZF_16QAM_VehB2}
\end{figure}

The following simulations first aim at demonstrating the accuracy of the derived asymptotic MSE expression in practical situations. Secondly, they validate the performance of the optimized precoders and decoders w.r.t. classical precoder and decoder designs. A FBMC-OQAM system is considered with $2M=128$ subcarriers and subcarrier spacing 15kHz as in LTE. The channels are randomly drawn from the ITU Vehicular A or B channel model, \textit{i.e.} a mildly or highly frequency selective channel respectively. {\color{black} Furthermore, they remain constant during the frame transmission (quasi-static assumption).} The Phydyas prototype pulse with overlapping factor $\kappa=4$ is used in the simulations \cite{bellanger2001specification}. This pulse does not fully satisfy the PR constraints but is of the near-perfect-reconstruction (NPR) type. Given that it almost fulfills PR constraints, the derived MSE expression \eref{eq:asymptotic_MSE} remains a very good approximation of the distortion, as will be shown in the following.

\fref{opt_decoder_MSE} shows the MSE of the classical and optimized MMSE decoders for a specific channel realization. The channel model simulated is the Vehicular B channel. The BS is assumed to have $N=4$ antennas serving $N_U=2$ users and the SNR of the system is 25dB. One can first check that the simulated MSE (in cross markers) perfectly matches the theoretical approximation (in solid line) of \eref{eq:asymptotic_MSE}. Furthermore, in the high SNR regime considered here, the classical MMSE decoder is limited by the distortion induced by the channel frequency selectivity. On the other hand, the optimized MMSE decoder uses the two extra antennas to cancel the distortion, giving a clear gain of performance.

{\color{black} In \fref{opt_decoder_SER_VehA} and \fref{fig:decoder_MMSE_ZF_16QAM_VehB2}, the symbol error rate (SER) for the classical and optimized decoders are plotted for a fixed number of users $N_U=3$ and different number of BS antennas $N$. The signal constellation is a 16-QAM. In \fref{opt_decoder_SER_VehA}, the Veh. A channel model is considered and the classical decoders can achieve the same performance as the optimized ones. This comes from the fact that the assumption of an approximately flat channel inside each subchannel is accurate. On the other hand, in \fref{fig:decoder_MMSE_ZF_16QAM_VehB2}, the Veh. B channel model is considered, i.e., a highly selective channel. The SER saturates very quickly with a classical decoder while the optimized ZF or MMSE decoder can compensate for the distortion as the number of BS antennas $N$ grows and the SER therefore saturates at higher SNR. In the case $N_U=3,N=6$, the SER does not even saturate in the considered SNR range since the BS has twice as many antennas and can completely remove the first order approximation of the distortion. This is in accordance with the asymptotic study at high SNR conducted in \sref{asymptotic_anal}. As expected, the MMSE designs outperform the ZF designs, and this gain is larger for a small number of BS antennas. Indeed, as the number of BS antennas increases, the interference can be better handled and the regularization gain of the MMSE decoder is reduced.
}

\begin{figure}[t!]
	\centering
		\includegraphics[width=4.5in]{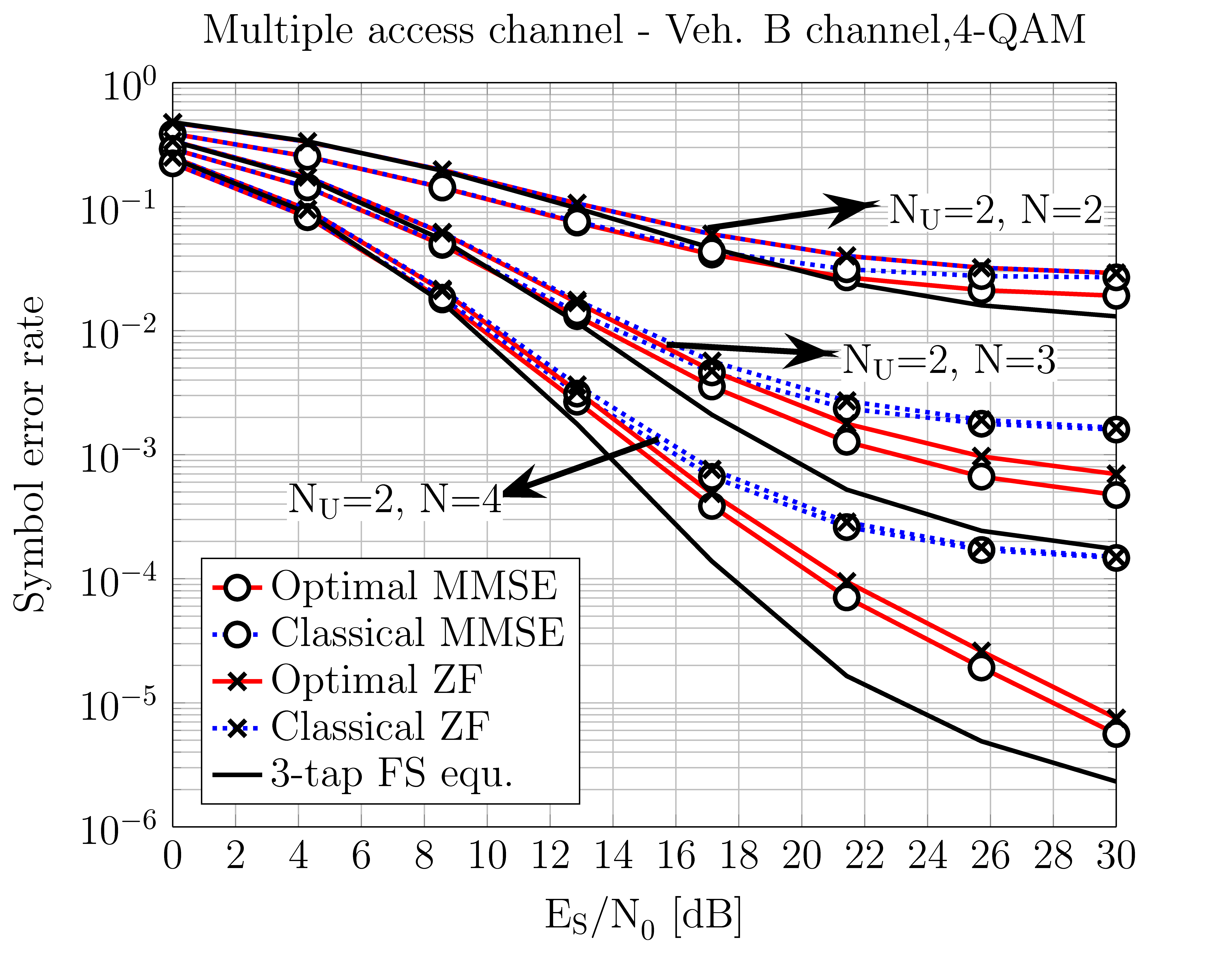}
%	\resizebox{0.45\textwidth}{!}{%
%		\Large
%		\input{fig/Opt_decoder_SER_VehB_MMSE_fs_comparison.tex}
%	}
	\caption{SER of the optimized and classical ZF/MMSE decoder and a 3-tap frequency sampling equalizer.}
	\label{opt_decoder_SER_VehB_MMSE}
\end{figure}

\begin{figure}[t!]
	\centering
		\includegraphics[width=4.5in]{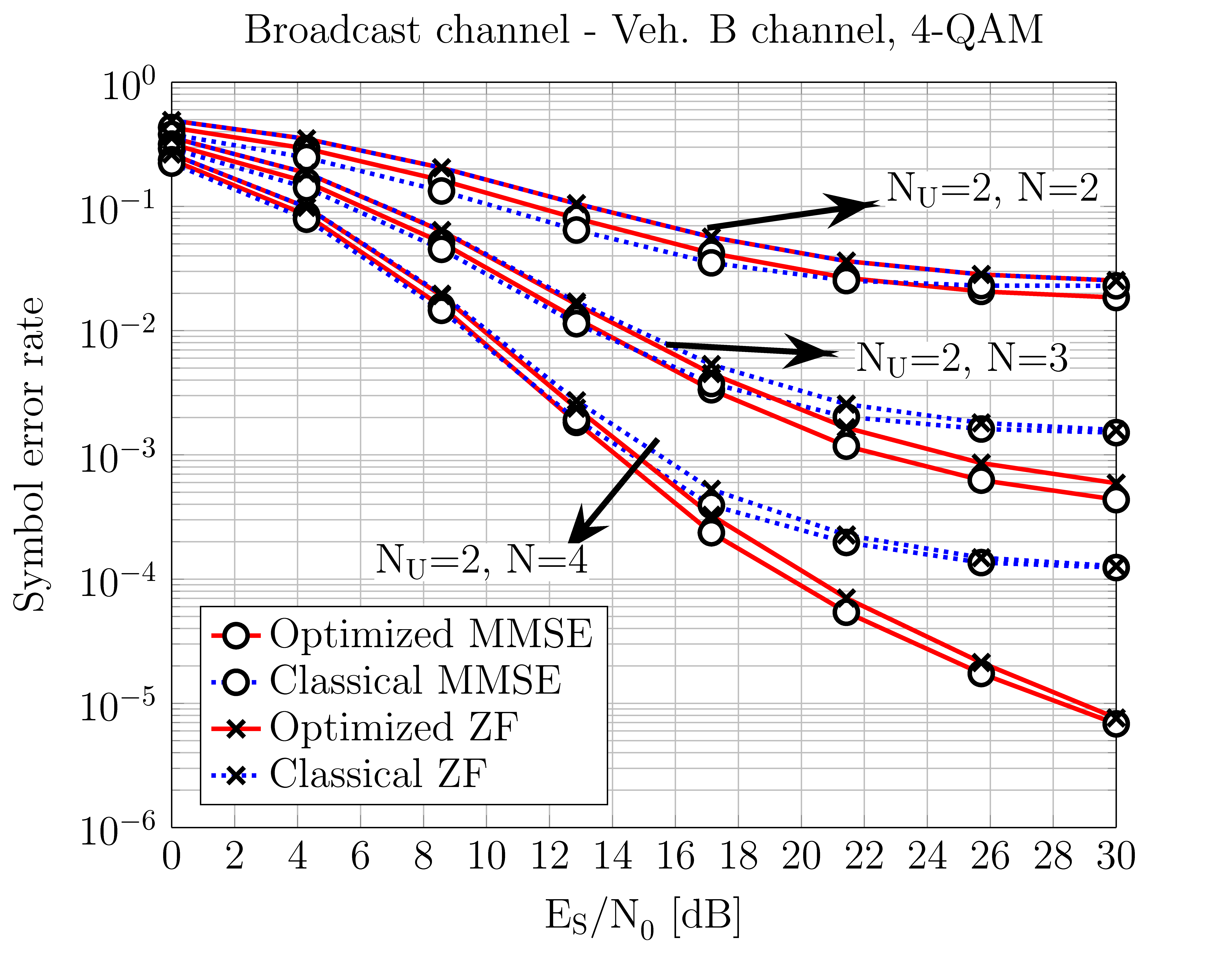}
%	\resizebox{0.45\textwidth}{!}{%
%		\Large
%		\input{fig/Opt_precoder_SER_VehB_MMSE2.tex}
%	}
	\caption{SER of the optimized, classical ZF/MMSE precoder.}
	\label{opt_precoder_SER_VehB_MMSE}
\end{figure}

{\color{black} In \fref{opt_decoder_SER_VehB_MMSE}, a 4-QAM constellation and the Veh. B channel model are considered in the UL. The proposed decoder designs are compared with a 3-tap frequency sampling equalizer that follows the design of \cite{ihalainen2011channel} with target frequency points chosen according to a ZF criterion. One can check that the 3-tap equalizer has a gain of performance relative to the proposed designs for the same antenna configuration. Note however that the multi-tap design has a much larger complexity in terms of hardware implementation and calculation of the equalizer coefficients. Moreover, it adds a reconstruction delay to the demodulation chain.}

{\color{black} \fref{opt_precoder_SER_VehB_MMSE} has exactly the same simulation parameters as \fref{opt_decoder_SER_VehB_MMSE} but for the DL. One can check that the performance of the ZF precoder in DL is similar to the decoder one in UL. Note that the SER performances of the precoder and decoder might differ. Indeed, even though the MSE expression is dual in UL and in DL, the distribution of the per-stream SER might differ due to the correlation of the noise arising in UL but not in DL. Moreover, due to the ZF constraint on $\Im(\HDL\mat{A}^{\mathrm{MMSE}}_{\mathrm{DL}})=0$, the optimized MMSE precoder performs slightly worse than the classical MMSE precoder at low SNR.}

\section{Conclusion}
\label{section_conclusion}

This paper investigated the design of optimized FBMC-OQAM precoders and decoders for a MU MIMO scenario highly selective channel. The asymptotic expression of the MSE of a FBMC transceiver was recalled, simplified and generalized. Optimizing the MSE expression, expressions of the optimized linear precoder and decoder were found under either a ZF criterion or a MMSE criterion. As soon as the BS has more antennas than the number of users, the optimized structures use those extra degrees of freedom to compensate for the distortion induced by channel frequency selectivity. From an asymptotic study at high SNR, it was shown that the first order approximation of the distortion can even be completely removed if the BS has at least twice as many antennas as the number of users. Simulation results have demonstrated the accuracy of the new asymptotic expression of the distortion as well as the performance of the optimized precoders and decoders.

% if have a single appendix:
%\appendix[Proof of the Zonklar Equations]
% or
%\appendix  % for no appendix heading
% do not use \section anymore after \appendix, only \section*
% is possibly needed

% use appendices with more than one appendix
% then use \section to start each appendix
% you must declare a \section before using any
% \subsection or using \label (\appendices by itself
% starts a section numbered zero.)
%

\appendices

\section{\label{sec:proof_distortion}Distortion expression}
{\color{black}In this appendix, we are interested only in the intrinsic distortion caused by the FBMC signal itself in the presence of channel frequency selectivity. Since the additive noise samples are assumed to be uncorrelated with the signal of interest, we will assume a noise-free received signal. In order to derive the distortion expression, we will first define different notations and recall one Lemma of \cite{Mestre16}.

We define $\vect{y}_{l_0,m_0}^{p,q} \in \mathbb{C}^{S\times 1}$ as
\begin{align*}
\vect{y}_{l_0,m_0}^{p,q}&=\sum_{l=0}^{2N_s-1} \sum_{m=0}^{2M-1} \vect{d}_{l,m} \sum_{n=0}^{L_q-1} p_{l,m}[n] q^*_{l_0,m_0}[n].
\end{align*}
The symbols $\vect{y}_{l_0,m_0}^{p,q}$ can be seen as the complex demodulated samples at subcarrier $m_0$ and multicarrier symbol $l_0$, before de-staggering, if the real-valued symbols streams $\vect{d}_{l,m}$ are FBMC/OQAM modulated using a prototype pulse $p[n]$ and demodulated using a prototype pulse $q[n]$ and for an ideal channel, i.e., $\mat{H}(\omega)=\mat{I}_S$. Moreover, if $p$ and $q$ are perfect reconstruction (bi-orthogonal) pulses, one will have $\vect{d}_{l_0,m_0}=\Re(\vect{y}_{l_0,m_0}^{p,q})$. 

As detailed in Section \ref{subsect:data_model}, to compensate for the effect of the channel, a single-tap precoding matrix $\mat{A}(\omega)$ and decoding matrix $\mat{B}\left(\omega\right)$ are used, operating at the per-subcarrier level. At the transmitter, the symbols $\vect{d}_{l,m}$ are precoded by matrix $\mat{A}(\omega_m)$. At the receiver, the equalized symbols are denoted by $\vect{x}_{l_0,m_0}^{p,q}=\mat{B}(\omega_{m_0})\vect{z}_{l_0,m_0}^{p,q}$ where $\vect{z}_{l_0,m_0}^{p,q}$ are the demodulated symbols at the receiver before decoding, i.e., 
\begin{align*}
	\vect{z}_{l_0,m_0}^{p,q}&=\sum_{n=0}^{L_q-1} \vect{r}[n] q^*_{l_0,m_0}[n].
\end{align*}
Note that one does not necessarily have $\Re(\vect{x}_{l_0,m_0}^{p,q}) =\vect{d}_{l_0,m_0}$ even if $\mat{B}\left(  \omega\right)  \mat{H}\left(  \omega\right) \mat{A}\left(  \omega\right)  =\mat{I}_S$ (and if no additive noise is present) due to the fact that the FBMC-OQAM orthogonality does not hold anymore if the channel is not exactly flat. However, this will hold approximately if all frequency-depending quantities, $\mat{B}\left(  \omega\right),\mat{H}\left(  \omega\right),\mat{A}\left(  \omega\right)$, are sufficiently flat as functions of $\omega$  provided that $\mat{B}\left(  \omega\right)  \mat{H}\left(  \omega\right) \mat{A}\left(  \omega\right)  =\mat{I}_S$. Our objective is to find an approximate expression for the associated error. This was already done in \cite{Mestre16} but for the special case of ZF (channel inversion), i.e., 
\[
\mat{B}\left(  \omega\right)  \mat{H}\left(  \omega\right)
\mat{A}\left(  \omega\right) =\mat{I}_{S}.
\]
We here extend and greatly simplify the formula derived in \cite{Mestre16} to the general case of non channel inversion, i.e.,
\[
\mat{B}\left(  \omega\right)  \mat{H}\left(  \omega\right)
\mat{A}\left(  \omega\right)  \neq\mat{I}_{S}.
\]

To derive the result of Theorem~\ref{Theorem}, we will use the following result, proven in \cite{Mestre16}. We will basically assume that all frequency depending quantities are smooth functions of $\omega$ and that the prototype pulses are sampled versions of smooth analog waveforms, namely, $\mathbf{(As1)}$-$\mathbf{(As3)}$. 

\begin{lemma}
	\label{lemma_Conversion_order}Under assumptions $\mathbf{(As1)}$%
	-$\mathbf{(As3)}$, we have
	\begin{align*}
		&\vect{z}_{l_0,m_0}^{p,q}\\
		&=\sum_{k=0}^{R}\sum_{k'=0}^{R-k}\frac{(-j)^{k+k'}}{(2M)^{k+k'}k!k'!}\mat{H}^{(k)} \mat{A}^{(k')} \vect{y}_{l_0,m_0}^{p,q^{(k)}}(k')		+ O(M^{-R})
	\end{align*}
	where
	\begin{align*}
		\vect{y}_{l_0,m_0}^{p,q^{(k)}}(k')&=\sum_{r=0}^{k'} \frac{k'!}{r!(k'-r)!} (-1)^r \vect{y}_{l_0,m_0}^{p^{(r)},q^{(k+k'-r)}}
	\end{align*}
	and where $p^{(r)},q^{(r)}$ are the sampled versions of the $r$-th time domain derivatives of the original prototype pulses. Superscript $^{(k)}$ denotes the m-th order derivative of the corresponding frequency dependent quantity, which always is evaluated at subcarrier $m_0$ (here and in the following of the appendix).
\end{lemma}

A direct application of the above lemma for $R=2$ allows us to write%
\begin{align}
	\vect{x}_{l_0,m_0}^{p,q}&=\mat{B}\mat{H}\mat{A}\vect{y}_{l_0,m_0} -\frac{j}{2M} \boldsymbol{\epsilon}_1-\frac{1}{2(2M)^2} \boldsymbol{\epsilon}_2+O(M^{-2}) \label{asympto_out}
\end{align}
where we have defined
\begin{align*}
\boldsymbol{\epsilon}_1&=\mat{B}\left(  \mat{HA}\right)  ^{(1)} \vect{y}_{l_0,m_0}^{p,q^{(1)}}-\mat{BHA}^{(1)} \vect{y}_{l_0,m_0}^{p^{(1)},q}\\
\boldsymbol{\epsilon}_2&=\mat{B}\left(  \mat{HA}\right)  ^{(2)} \vect{y}_{l_0,m_0}^{p,q^{(2)}} - 2 \mat{B}\left(  \mat{HA}^{(1)}\right)  ^{(1)} \vect{y}_{l_0,m_0}^{p^{(1)},q^{(1)}} \\
&+\mat{BHA}^{(2)}\vect{y}_{l_0,m_0}^{p^{(2)},q}
\end{align*}
At this point, force $p=q$ so that the same prototype pulse is used at both
transmitter and receiver. We define the distortion error associated with the
complex-valued symbols as\footnote{Note that this definition is in accordance with equation (\ref{eq:asymptotic_MSE}) when noise is absent.}
\begin{align}
	P_{d}(m)=2\mathbb{E} \left( \|\Re(\vect{x}_{l_0,m_0}^{p,q}) -\vect{d}_{l_0,m_0}\|^{2}\right) \label{distortion_power} .
\end{align}

We now need to define some pulse-related quantities. Given two generic pulses, $p,q$ of
length $2M\kappa$ and let $\mat{P}$ and $\mat{Q}$ denote two
$2M\times\kappa$ matrices obtained by arranging the samples of the respective
pulses in columns from left to right. We will define
\begin{align*}
\mathcal{R}(p,q)  &  =\mat{P\circledast J}_{2M}\mat{Q}\\
\mathcal{S}(p,q)  &  =\left(  \mat{J}_{2}\otimes\mat{I}_{M}\right)
\mat{P\circledast J}_{2M}\mat{Q}%
\end{align*}
where $\mathbf{\circledast}$ denotes row-wise convolution, $\otimes$ denotes
Kronecker product, $\mat{I}_{M}$ (resp. $\mat{J}_{M}$) are the identity (resp. exchange) matrices of order $M$. Given four generic pulses $p,q,r,s$, we define
\begin{align*}
&\eta^{\pm}\left(  p,q,r,s\right)  \\
&  =\frac{1}{2M}\text{tr}\left[
\mat{U}^{+}\mathcal{R}\left(  p,q\right)  \mathcal{R}^{T}\left(
r,s\right)  +\mat{U}^{-}\mathcal{S}\left(  p,q\right)  \mathcal{S}%
^{T}\left(  r,s\right)  \right]  \\
&\eta^{\mp}\left(  p,q,r,s\right)   \\
&  =\frac{1}{2M}\text{tr}\left[
\mat{U}^{-}\mathcal{R}\left(  p,q\right)  \mathcal{R}^{T}\left(
r,s\right)  +\mat{U}^{+}\mathcal{S}\left(  p,q\right)  \mathcal{S}%
^{T}\left(  r,s\right)  \right]
\end{align*}
where $\mat{U}^{\mat{\pm}}=\mat{I}_{2}\otimes\left(  \mat{I}%
_{M}\mat{\pm J}_{M}\right)  $. In order to simplify the notations, given four integers $m,n,r,s$, we will define $\eta_{mnrs}%
^{(+,-)}=\eta^{\pm}\left(  p^{(m)},p^{(n)},p^{(r)},p^{(s)}\right)  $.

Assuming that the pulse has PR properties $\mathbf{(As2)}$ so that $\vect{d}_{l_0,m_0}=\Re(\vect{y}_{l_0,m_0}^{p,q})$, we can
obtain an asymptotic expression for this distortion by simply inserting
(\ref{asympto_out}) in (\ref{distortion_power}) for $p=q$ and by using the fact
that \cite[Appendix B]{mestre13tsp}%
\begin{align*}
\mathbb{E}  \left(\Re(\vect{y}_{l_0,m_0}^{p^{(m)},q^{(n)}})\Re^T(\vect{y}_{l_0,m_0}^{p^{(r)},q^{(s)}})\right)  &  =\eta_{mnrs}^{(+,-)}\mat{I}_S\\
\mathbb{E}  \left(\Im(\vect{y}_{l_0,m_0}^{p^{(m)},q^{(n)}})\Im^T(\vect{y}_{l_0,m_0}^{p^{(r)},q^{(s)}})\right)  &  =\eta_{mnrs}^{(-,+)}\mat{I}_S\\
\mathbb{E}  \left(\Re(\vect{y}_{l_0,m_0}^{p^{(m)},q^{(n)}})\Im^T(\vect{y}_{l_0,m_0}^{p^{(r)},q^{(s)}})\right)  &  =\mat{0}.
\end{align*}}
The resulting expression can be more compactly expressed by using the fact
that $\eta_{mnrs}^{(+,-)}=\eta_{rsmn}^{(+,-)}$, $\eta_{mnrs}^{(+,-)}%
=\eta_{nmsr}^{(+,-)}$,\footnote{Obviously, the same identities hold if $(+,-)$ is replaced by
	$(-,+)$.} as established in Lemma \ref{lemma:non_sym_pulses}\ of
Appendix \ref{sec:eta_properties}. Furthermore, if we assume that the
prototype pulse is either symmetric or anti-symmetric, one can establish that
$\eta_{0000}^{(+,-)}=\eta_{0000}^{(-,+)}$, $\eta_{0001}^{(+,-)}=\eta
_{0001}^{(-,+)}$, $\eta_{0101}^{(+,-)}=\eta_{0101}^{(-,+)}$ and $\eta
_{0020}^{(+,-)}=\eta_{0020}^{(-,+)}$, a fact that is proven in Lemma
\ref{lemma:etaplusminus_minusplus} of Appendix \ref{sec:eta_properties}. If,
additionally, the prototype pulse meets the PR conditions, we can guarantee
that $\eta_{0001}^{(+,-)}=0$ as established in Lemma \ref{lemma:exact_PR} of
Appendix \ref{sec:eta_properties}.

Using all this, together with the fact that%
\begin{align*}
\text{tr}\left[  \Re(\mat{X})\Re^{T}%
(\mat{Y})+\Im(\mat{X})\Im^{T}(\mat{Y})\right]
&  =\Re  \text{tr}\left[  \mat{XY}^{H}\right]  \\
\text{tr}\left[  \Im(\mat{X})\Re^{T}%
(\mat{Y})-\Re(\mat{X})\Im^{T}(\mat{Y})\right]
&  =\Im \text{tr}\left[  \mat{XY}^{H}\right] 
\end{align*}
for any complex valued matrices $\mat{X}$, $\mat{Y}$ of appropriate
dimensions, we see that
\begin{equation}
P_{d}(m)=2\xi_{0,m}+\frac{2}{\left(  2M\right)  ^{2}}\xi_{2,m}+O\left(
M^{-2}\right)  \label{eq:pe1_exact}%
\end{equation}
where we have defined
\begin{align*}
\xi_{0,m}  &  =\eta_{0000}^{(+,-)}\text{tr}\left[  \left(  \mat{BHA-I}%
\right)  \left(  \mat{BHA-I}\right)  ^{H}\right] \\
\xi_{2,m}  &  =\eta_{1010}^{(+,-)}\left(  \text{tr}\left[  \mat{BHA}%
^{(1)}\left(  \mat{BHA}^{(1)}\right)  ^{H}\right] +\text{tr}\left[
\mat{B}\left(  \mat{HA}\right)  ^{(1)}\left(  \mat{B}\left(
\mat{HA}\right)  ^{(1)}\right)  ^{H}\right]  \right) \\
&  -\eta_{2000}^{(+,-)}\left(  \Re\text{tr}\left[  \left(
\mat{BHA-I}\right)  \left(  \mat{BHA}^{(2)}\right)  ^{H}\right]
 +\Re\text{tr}\left[  \left(  \mat{BHA-I}\right)  \left(
\mat{B}\left(  \mat{HA}\right)  ^{(2)}\right)  ^{H}\right]  \right) \\
&  +2\eta_{0011}^{(+,-)}\text{tr}\left[  \Re\left[
\mat{BHA-I}\right]  \Re^{T}\left[  \mat{B}\left(
\mat{HA}^{(1)}\right)  ^{(1)}\right]  \right] \\
&+2\eta_{0011}^{(-,+)}%
\text{tr}\left[  \Im\left[  \mat{BHA-I}\right]
\Im^{T}\left[  \mat{B}\left(  \mat{HA}^{(1)}\right)
^{(1)}\right]  \right] \\
&  -2\eta_{1001}^{(+,-)}\text{tr}\left[  \Im\left[
\mat{B}\left(  \mat{HA}\right)  ^{(1)}\right]  \Im%
^{T}\left[  \mat{BHA}^{(1)}\right]  \right] \\
&-2\eta_{1001}^{(-,+)}%
\text{tr}\left[  \Re\left[  \mat{B}\left(  \mat{HA}%
\right)  ^{(1)}\right]  \Re^{T}\left[  \mat{BHA}%
^{(1)}\right]  \right]
\end{align*}
and where all frequency-depending matrices are evaluated at $\omega=\omega
_{m}$.

It is easy to see that, thanks to the PR property of the prototype pulse, we
will have $\eta_{0000}^{(+,-)}=P_{s}/2$, where we recall that $P_{s}$ is the
power of the complex constellation\ symbols. Therefore, the asymptotic
distortion power depends on six different pulse-related quantities, which
present some non-trivial interrelationships. To establish these
interrelationships, we invoke again Lemma \ref{lemma:etaplusminus_minusplus}
of Appendix \ref{sec:eta_properties}. In particular, the relationships in
(\ref{eq:eta_change_signs_1})-(\ref{eq:eta_change_signs_2}) allow us to
establish
\begin{gather}
\left[  \eta_{1001}^{(+,-)}-\eta_{1001}^{(-,+)}\right]  -\left[  \eta
_{0011}^{(+,-)}-\eta_{0011}^{(-,+)}\right]  =0\label{eq:rel_eta_1}\\
\eta_{0011}^{(+,-)}+\eta_{1010}^{(+,-)}=0. \label{eq:rel_eta_2}%
\end{gather}
We can find further equivalences between the different pulse quantities by
considering here the asymptotic domain as $M\rightarrow\infty$ together with
the fact that the prototype pulse is, by assumption, a sampled version of a
smooth analog waveform.

Indeed, applying the result in Proposition \ref{proposition:etas_asym_deriv},
Corollary \ref{corollary:etas_asym_deriv_PR} and Lemma
\ref{lemma:etas_asym_deriv_PR} in Appendix \ref{sec:eta_properties} together
with the fact that the prototype pulse is symmetric and PR compliant, we
obtain%
\begin{align}
\eta_{2000}^{(+,-)}-\eta_{0011}^{(+,-)}+\eta_{1010}^{(+,-)}-\eta
_{1001}^{(-,+)}  &  =O(M^{-1})\label{eq:rel_eta_3}\\
\eta_{0011}^{(+,-)}-\eta_{2000}^{(+,-)}  &  =O(M^{-2}). \label{eq:rel_eta_4}%
\end{align}
We may consider the system of equations formed by (\ref{eq:rel_eta_1}%
),(\ref{eq:rel_eta_2}),(\ref{eq:rel_eta_3}) and (\ref{eq:rel_eta_4}). Denoting
$\alpha=\eta_{1010}^{(+,-)}$ and $\beta=\eta_{0011}^{(-,+)}$, we can express
the system of equations as%
\[
\left[
\begin{array}
[c]{cccc}%
1 & -1 & -1 & 0\\
0 & 0 & 1 & 0\\
0 & -1 & -1 & 1\\
0 & 0 & 1 & -1
\end{array}
\right]  \left[
\begin{array}
[c]{c}%
\eta_{1001}^{(+,-)}\\
\eta_{1001}^{(-,+)}\\
\eta_{0011}^{(+,-)}\\
\eta_{2000}^{(+,-)}%
\end{array}
\right]  =\left[
\begin{array}
[c]{c}%
-\beta\\
-\alpha\\
-\alpha\\
0
\end{array}
\right]  +O(M^{-1})
\]
so that we can conclude that
\[
\left[
\begin{array}
[c]{c}%
\eta_{1001}^{(+,-)}\\
\eta_{1001}^{(-,+)}\\
\eta_{0011}^{(+,-)}\\
\eta_{2000}^{(+,-)}%
\end{array}
\right]  =\left[
\begin{array}
[c]{c}%
-\beta\\
\alpha\\
-\alpha\\
-\alpha
\end{array}
\right]  +O(M^{-1})
\]
Using this in (\ref{eq:pe1_exact}) we obtain (\ref{eq:asymptotic_distortion}).

\section{\label{sec:eta_properties}Some properties of the quantities
	$\eta_{mnrs}^{(+,-)}$}

In this appendix, we provide some\ identities on the quantities $\eta
_{mnrs}^{(+,-)}$, $\eta_{mnrs}^{(-,+)}$ that will clearly simplify the
expression for the asymptotic distortion derived above. We will begin by presenting some
properties that hold exactly for all values of $M$.

\subsection{Non-asymptotic properties}

Let us begin with general properties that hold regardless of whether the pulses
are symmetric or not.

\begin{lemma}
	\label{lemma:non_sym_pulses}By the definition of the $\eta^{\pm}\left(
	p,q,r,s\right)  $, and regardless of the pulse symmetries, we have
	\[
	\eta^{\pm}\left(  p,q,r,s\right)  =\eta^{\pm}\left(  r,s,p,q\right)
	\]
	and
	\[
	\eta^{\pm}\left(  p,q,r,s\right)  =\eta^{\pm}\left(  q,p,s,r\right)  .
	\]
	The same identities hold if $\pm$ is replaced by $\mp$ everywhere.
\end{lemma}

\begin{proof}
	Indeed, the first result is a consequence of the fact that
	\begin{align*}
	\left(  \mat{U}^{\mat{\pm}}\mathcal{R}\left(  p,q\right)  \mathcal{R}%
	^{T}\left(  r,s\right)  \right)  ^{T}  &  =\mathcal{R}\left(  r,s\right)
	\mathcal{R}^{T}\left(  p,q\right)  \mat{U}^{\mat{\pm}}\\
	\left(  \mat{U}^{\mat{\pm}}\mathcal{S}\left(  p,q\right)  \mathcal{S}%
	^{T}\left(  r,s\right)  \right)  ^{T}  &  =\mathcal{S}\left(  r,s\right)
	\mathcal{S}^{T}\left(  p,q\right)  \mat{U}^{\mathbf{\pm}}%
	\end{align*}
	whereas the second one follows from the identities
	\begin{align*}
	\mathcal{R}\left(  p,q\right)  \mathcal{R}^{T}\left(  r,s\right)   &
	=\mat{J}_{2M}\mathcal{R}\left(  q,p\right)  \mathcal{R}^{T}\left(
	s,r\right)  \mat{J}_{2M},\\
	\mathcal{S}\left(  p,q\right)  \mathcal{S}^{T}\left(  r,s\right)   &  =\left(
	\mat{I}_{2}\otimes\mat{J}_{M}\right)  \mathcal{S}\left(  q,p\right)
	\mathcal{S}^{T}\left(  s,r\right)  \left(  \mat{I}_{2}\otimes\mat{J}%
	_{M}\right)
	\end{align*}
	and the definition of $\mat{U}^{\mat{\pm}}$.
\end{proof}

Sometimes, it is useful to consider a relationship between quantities of the
type $\eta^{\pm}$ and $\eta^{\mp}$. In order to obtain such relationships, we
impose that the pulses are either symmetric or anti-symmetric in the time domain.

\begin{lemma}
	\label{lemma:etaplusminus_minusplus}Assume that all the pulses $p,q,r,s$ are
	either symmetric or anti-symmetric in the time domain. Let $s(p)$ be defined
	so that $s(p)=0$ if the pulse $p$ has even symmetry and $s(p)=1$ if the pulse
	is anti-symmetric. Then, we can write%
	\begin{align}
	&\left(  -1\right)  ^{s(p)}\left[  \eta^{\pm}(p,q,r,s)-\eta^{\mp}%
	(p,q,r,s)\right] \nonumber \\
	& +\left(  -1\right)  ^{s(r)}\left[  \eta^{\pm}(r,q,p,s)-\eta
	^{\mp}(r,q,p,s)\right]  =0 \label{eq:eta_change_signs_1}%
	\end{align}
	and also
	\begin{align}
	&\left(  -1\right)  ^{s(p)}\left[  \eta^{\pm}(p,q,r,s)+\eta^{\mp}%
	(p,q,r,s)\right] \nonumber \\
	& \pm\left(  -1\right)  ^{s(p)}\left[  \eta^{\pm
	}(p,q,s,r)-\eta^{\mp}(p,q,s,r)\right]  \nonumber\\
	&=\left(  -1\right)  ^{s(s)}\left[  \eta^{\pm}(s,q,r,p)+\eta^{\mp
	}(s,q,r,p)\right]  \label{eq:eta_change_signs_2}%
	\end{align}
	In particular, for the specific case where $p=r$ we have%
	\[
	\eta^{\pm}(p,q,p,s)=\eta^{\mp}(p,q,p,s).
	\]
	Finally, if $q=s$ and the pulses $p$ and $r$ have the same type of symmetry,
	we have
	\[
	\eta^{\pm}(p,q,r,q)=\eta^{\mp}(p,q,r,q).
	\]
	
\end{lemma}

\begin{proof}
	Let us denote by $\mathcal{R}_{1}(p,q)$ and $\mathcal{R}_{2}(p,q)$ the upper
	and lower matrices of $\mathcal{R}(p,q)$, and equivalently for $\mathcal{S}%
	_{1}(p,q)$ and $\mathcal{S}_{2}(p,q)$. Let $\mat{P}$ denote a
	$2M\times\kappa$ matrix obtained by arranging the pulse $p[n]$ in columns, and
	let $\mat{P}_{1}$ and $\mat{P}_{2}$ respectively denote the matrices
	obtained by selecting the $M$ upper and lower rows of $\mat{P}$
	respectively. By the symmetry of $p[n]$, we know that
	\[
	\mat{P}_{1}=(-1)^{s(p)}\mat{J}_{M}\mat{P}_{2}\mat{J}_{2\kappa-1}.
	\]
	On the other hand, we can prove that, for any four matrices $\mat{A,B,C}$
	and $\mat{D}$ of dimensions $M\times\kappa$, the diagonal entries of
	$\left(  \mat{A}\circledast\mat{BJ}_{\kappa}\right)  \left(
	\mat{C}\circledast\mat{DJ}_{\kappa}\right)  ^{T}\mat{J}_{M}$ are
	equal to the diagonal entries of $\mat{J}_{M}\left(  \mat{C}%
	\circledast\mat{J}_{M}\mat{B}\right)  \left(  \mat{A}\circledast
	\mat{J}_{M}\mat{D}\right)  ^{T}$. This shows that, using pulse symmetry,%
	\begin{align}
	\text{tr}\left[  \mathcal{R}_{1}\left(  p,q\right)  \mathcal{R}_{1}^{T}\left(
	r,s\right)  \mat{J}_{M}\right] 	 &  =\text{tr}\left[  \mat{J}_{M}\left(
	\mat{J}_{M}\mat{P}_{1}\circledast\mat{Q}_{2}\right)  \left(
	\mat{J}_{M}\mat{R}_{1}\circledast\mat{S}_{2}\right)  ^{T}\right]
	\nonumber\\
	&  =\left(  -1\right)  ^{s(p)+s(r)}\text{tr}\left[  \left(  \mat{P}%
	_{2}\mat{J}_{\kappa}\circledast\mat{Q}_{2}\right)  \left(
	\mat{R}_{2}\mat{J}_{\kappa}\circledast\mat{S}_{2}\right)
	^{T}\mat{J}_{M}\right]  \nonumber\\
	&  \overset{(\ast)}{=}\left(  -1\right)  ^{s(p)+s(r)}\text{tr}\left[
	\mat{J}_{M}\left(  \mat{S}_{2}\circledast\mat{J}_{M}\mat{P}%
	_{2}\right)  \left(  \mat{Q}_{2}\circledast\mat{J}_{M}\mat{R}%
	_{2}\right)  ^{T}\right]  \nonumber\\
	&  =\left(  -1\right)  ^{s(p)+s(r)}\text{tr}\left[  \mat{J}_{M}\left(
	\mat{P}_{2}\circledast\mat{J}_{M}\mat{S}_{2}\right)  \left(
	\mat{R}_{2}\circledast\mat{J}_{M}\mat{Q}_{2}\right)  ^{T}\right]
	\nonumber\\
	&  =\left(  -1\right)  ^{s(p)+s(r)}\text{tr}\left[  \mat{J}_{M}%
	\mathcal{S}_{1}\left(  p,s\right)  \mathcal{S}_{1}^{T}\left(  r,q\right)
	\right]  \nonumber\\
	&  =\left(  -1\right)  ^{s(p)+s(r)}\text{tr}\left[  \mat{J}_{M}%
	\mathcal{S}_{1}\left(  r,q\right)  \mathcal{S}_{1}^{T}\left(  p,s\right)
	\right]  \label{eq:traceR1R1J}
	\end{align}
	where the identity in ($\ast$) follows from the above convolution result.
	Equivalently, we will obviously have%
	\begin{align*}
	\text{tr}\left[  \mathcal{R}_{2}\left(  p,q\right)  \mathcal{R}_{2}^{T}\left(
	r,s\right)  \mat{J}_{M}\right]=\left(  -1\right)  ^{s(p)+s(r)}%
	\text{tr}\left[  \mat{J}_{M}\mathcal{S}_{2}\left(  r,q\right)
	\mathcal{S}_{2}^{T}\left(  p,s\right)  \right]
		\end{align*}
	and the two above identities directly prove (\ref{eq:eta_change_signs_1}).
	Regarding the identity in (\ref{eq:eta_change_signs_2}), it follows easily
	from the above identities together with the fact that $\mat{J}%
	_{M}\mathcal{R}_{i}\left(  p,q\right)  =\mathcal{R}_{3-i}\left(  q,p\right)
	$, $\mat{J}_{M}\mathcal{S}_{i}\left(  p,q\right)  =\mathcal{S}_{i}\left(
	q,p\right)  $, $i=1,2$, and the fact that
	\begin{align*}
	\text{tr}\left[  \mathcal{R}_{1}\left(  p,q\right)  \mathcal{R}_{2}^{T}\left(
	r,s\right)  \mat{J}_{M}\right]   & =\left(  -1\right)  ^{s(p)+s(r)}\text{tr}\left[  \mathcal{R}_{1}\left(  r,q\right)  \mathcal{R}_{2}^{T}\left(
	p,s\right)  \mat{J}_{M}\right]  \\
	\text{tr}\left[  \mathcal{S}_{1}\left(  p,q\right)  \mathcal{S}_{2}^{T}\left(
	r,s\right)  \mat{J}_{M}\right]  &=\left(  -1\right)  ^{s(p)+s(r)}%
	\text{tr}\left[  \mathcal{S}_{1}\left(  r,q\right)  \mathcal{S}_{2}^{T}\left(
	p,s\right)  \mat{J}_{M}\right]
	\end{align*}
	which can be established following the same approach as in
	(\ref{eq:traceR1R1J}). The last two identities in the statement of the lemma
	are obtained as special cases of (\ref{eq:eta_change_signs_1}).
\end{proof}

We finalize the description of the non-asymptotic properties of the $\eta
^{\pm}(p,q,r,s)$ with a result that will be useful whenever two of the pulses
meet the perfect reconstruction conditions.

\begin{lemma}
	\label{lemma:exact_PR}Assume that the two pulses $p,q$ meet the perfect
	reconstruction conditions, and that $r$ and $s$ are either symmetric or
	anti-symmetric in the time domain and have the opposite symmetry. Then
	$\eta^{\pm}(p,q,r,s)=0$.
\end{lemma}

\begin{proof}
	Since $p,q$ have PR conditions, we know that $\mat{U}^{-}\mathcal{S}\left(
	p,q\right)  $ is an all-zero matrix whereas $\mat{U}^{+}\mathcal{R}\left(
	p,q\right)  $ has zeros everywhere except for the central column, which is
	filled with $1$s. Therefore, we are able to write
	\[
	\eta^{\pm}(p,q,r,s)=\frac{1}{2M}\sum_{n=1}^{2M\kappa}r[n]s[2M\kappa-n+1]=0
	\]
	where the last equality follows from the fact that $r$ and $s$ have the
	opposite symmetry.
\end{proof}

\subsection{Asymptotic properties}

Let us now consider some properties of the $\eta^{\pm}(p,q,p,s)$ that are
obtained by assuming that the pulses are sampled versions of a smooth analog
waveform. In other words, we assume that $p[n]$, $q[n]$, $r[n]$ and $s[n]$ are
sampled versions of the waveforms $p(t),q(t),r(t)$ and $s(t)$ respectively,
according to the properties in $(\mathbf{As2})$. This means that we can
express
\[
p[n]=p\left(  \left(  n-\frac{N+1}{2}\right)  \frac{T_{s}}{2M}\right)
\]
where $p\left(  t\right)  $ has the usual properties in $(\mathbf{As2})$. The
same holds for the rest of the pulses.

Let us denote $p_{m}[n]$ the $m$-th polyphase component of $p[n]$, which can be
expressed as
\[
p_{m}[n]=p_{m}\left(  \left(  n-\frac{1}{2}\right)  \frac{T_{s}}{2M}\right)
\]
where $p_{m}(t)$ is the $m$-th section of $p(t)$, namely
\[
p_{m}(t)=p\left(  t-\left(  m+\frac{\kappa}{2}-1\right)  T_{s}\right)
\]
which has support $[0,T_{s}]$. The definition of $p_{m}(t)$ is only valid for
$m=1,\ldots,\kappa$, but we will consider $p_{m}(t)=0$ for values of $m$
outside this range. The same definitions carry over to the other pulses,
namely $q,r$ and $s$.

With all these definitions, we are now in a position to establish the first
asymptotic result associated with $\eta^{\pm}(p,q,p,s)$. The following result
asymptotically relates the original quantity $\eta^{\pm}(p,q,p,s)$ with an
equivalent definition that is constructed using the analog waveforms instead
of the sampled ones.

\begin{lemma}
	\label{lemma:eta_to_integral}Under the above assumptions, we can write
	\[
	\eta^{\pm}(p,q,r,s)=\bar{\eta}^{\pm}(p,q,r,s)+O(M^{-2})
	\]
	where
	\begin{align*}
		\bar{\eta}^{\pm}(p,q,r,s)	&=\sum_{\ell=1}^{2\kappa-1}\sum_{m=1}^{\kappa}%
		\sum_{n=1}^{\kappa}A^{\left(  \ell,m,n\right)  }\left[  p,q,r,s\right]  \pm
		B^{\left(  \ell,m,n\right)  }\left[  p,q,r,s\right]
	\end{align*}
	
	and where $A^{\left(  \ell,m,n\right)  }\left[  p,q,r,s\right]  $ and
	$B^{\left(  \ell,m,n\right)  }\left[  p,q,r,s\right]  $ are defined as:
	\begin{align*} 
	A^{\left(  \ell,m,n\right)  }\left[  p,q,r,s\right]  &  =\frac{1}{T_{s}}%
	\int_{0}^{T_{s}}p_{m}(t)q_{\ell-m+1}(T_{s}-t)r_{n}(t)s_{\ell-n+1}(T_{s}-t)dt\\
	&  +\frac{1}{T_{s}}\int_{0}^{\frac{T_s}{2}}p_{m}(t)q_{\ell-m+1}(\frac{T_s}{2}-t)r_{n}%
	(t)s_{\ell-n+1}(\frac{T_s}{2}-t)dt\\
	&  +\frac{1}{T_{s}}\int_{\frac{T_s}{2}}^{T_{s}}p_{m}(t)q_{\ell-m+1}(\frac{3T_s}{2}-t)r_{n}(t)s_{\ell-n+1}(\frac{3T_s}{2}-t)dt
	\end{align*}
	and
	\begin{align*}
	B^{\left(  \ell,m,n\right)  }\left[  p,q,r,s\right]  &  =\frac{1}{T_{s}}%
	\int_{0}^{\frac{T_s}{2}}p_{m}(t)q_{\ell-m+1}(T_{s}-t)r_{n}(\frac{T_s}{2}-t)s_{\ell
		-n+1}(\frac{T_s}{2}+t)dt\\
	&  +\frac{1}{T_{s}}\int_{\frac{T_s}{2}}^{T_{s}}p_{m}(t)q_{\ell-m+1}(T_{s}%
	-t)r_{n}(\frac{3T_s}{2}-t)s_{\ell-n+1}(t-\frac{T_s}{2})dt\\
	&  +\frac{1}{T_{s}}\int_{0}^{\frac{T_s}{2}}p_{m}(t)q_{\ell-m+1}(\frac{T_s}{2}-t)r_{n}%
	(\frac{T_s}{2}-t)s_{\ell-n+1}(t)dt\\
	&  +\frac{1}{T_{s}}\int_{\frac{T_s}{2}}^{T_{s}}p_{m}(t)q_{\ell-m+1}(\frac{3T_s}{2}-t)r_{n}(\frac{3T_s}{2}-t)s_{\ell-n+1}(t)dt.
	\end{align*}
	
\end{lemma}

\begin{proof}
	The proof is a direct consequence of the definition of $\eta^{\pm}(p,q,r,s)$
	and the Riemann integral. Details are omitted due to the space constraints.
\end{proof}

The above lemma allows us to express $\eta^{\pm}(p,q,r,s)$ as a function of
integrals of the analog waveform sections $p_{m}(t),q_{m}(t),r_{m}%
(t),s_{m}(t)$. This turns out to be very convenient for the following result,
which provides an asymptotic relationship among different $\eta^{\pm
}(p,q,r,s)$ with respect to the derivatives of the corresponding pulses.

\begin{proposition}
	\label{proposition:etas_asym_deriv}Under the above assumptions and
	definitions, we can write
	\begin{align}
	\eta^{\pm}\left(  p^{\prime},q,r,s\right)  -\eta^{\pm}\left(  p,q^{\prime
	},r,s\right)  +\eta^{\mp}\left(  p,q,r^{\prime},s\right)-\eta^{\mp}\left(
	p,q,r,s^{\prime}\right) 	&=-\left[  \mat{U}^{\mp}\mathcal{R}\left(  p,q\right)  \mathcal{R}%
	^{T}\left(  r,s\right)  \mat{U}^{\pm}\right]  _{1,1} \nonumber\\
	&-\left[  \mat{U}%
	^{\mp}\mathcal{R}\left(  p,q\right)  \mathcal{R}^{T}\left(  r,s\right)
	\mat{U}^{\pm}\right]  _{M+1,M+1} \nonumber\\
	&-\left[  \mat{U}^{\pm}\mathcal{S}\left(  p,q\right)  \mathcal{S}^{T}\left(
	r,s\right)  \mat{U}^{\mp}\right]  _{1,1} \nonumber\\
	&-\left[  \mat{U}^{\pm
	}\mathcal{S}\left(  p,q\right)  \mathcal{S}^{T}\left(  r,s\right)
	\mat{U}^{\mp}\right]  _{M+1,M+1}+O(M^{-1}) \label{eq:eta_asym_derivs}%
	\end{align}
	
\end{proposition}

\begin{proof}
	Consider the definition of $A^{\left(  \ell,m,n\right)  }\left[
	p,q,r,s\right]  $ and $B^{\left(  \ell,m,n\right)  }\left[  p,q,r,s\right]  $.
	These terms consist of a number of integrals of a differentiable function on a
	compact interval of the positive real axis. Hence, we can use the
	fundamental theorem of calculus to write%
	\begin{align*}
	&\bar{\eta}^{\pm}\left(  p^{\prime},q,r,s\right)  -\bar{\eta}^{\pm}\left(
	p,q^{\prime},r,s\right)  +\bar{\eta}^{\mp}\left(  p,q,r^{\prime},s\right)
	-\bar{\eta}^{\mp}\left(  p,q,r,s^{\prime}\right)=\sum_{\ell=1}^{2\kappa
		-1}\phi_{\ell}%
	\end{align*}
	where%
	\begin{align*}
	\phi_{\ell} &  =\mu_{\ell}\left(  T_{s},0\right)  \xi_{\ell}\left(
	T_{s},0\right)  -\mu_{\ell}\left(  0,T_{s}\right)  \xi_{\ell}\left(
	0,T_{s}\right)  \\
	&  \pm\left[  \mu_{\ell}\left(  T_{s},0\right)  -\mu_{\ell}\left(
	0,T_{s}\right)  \right]  \xi_{\ell}\left(  T_{s}/2,T_{s}/2\right)  \\
	&  \pm\mu_{\ell}\left(  T_{s}/2,T_{s}/2\right)  \left[  \xi_{\ell}\left(
	0,T_{s}\right)  -\xi_{\ell}\left(  T_{s},0\right)  \right]  \\
	&  +\mu_{\ell}\left(  T_{s}/2,0\right)  \xi_{\ell}\left(  T_{s}/2,0\right)
	-\mu_{\ell}\left(  0,T_{s}/2\right)  \xi_{\ell}\left(  0,T_{s}/2\right)  \\
	&  +\mu_{\ell}\left(  T_{s},T_{s}/2\right)  \xi_{\ell}\left(  T_{s}%
	,T_{s}/2\right)  -\mu_{\ell}\left(  T_{s}/2,T_{s}\right)  \xi_{\ell}\left(
	T_{s}/2,T_{s}\right)  \\
	&  \mp\mu_{\ell}\left(  T_{s}/2,0\right)  \xi_{\ell}\left(  0,T_{s}/2\right)
	\pm\mu_{\ell}\left(  0,T_{s}/2\right)  \xi_{\ell}\left(  T_{s}/2,0\right)  \\
	&  \mp\mu_{\ell}\left(  T_{s},T_{s}/2\right)  \xi_{\ell}\left(  T_{s}%
	/2,T_{s}\right)  \pm\mu_{\ell}\left(  T_{s}/2,T_{s}\right)  \xi_{\ell}\left(
	T_{s},T_{s}/2\right)
	\end{align*}
	and where we have defined
	\begin{align*}
	\mu_{\ell}\left(  t_{1},t_{2}\right)   &  =\sum_{m=1}^{\kappa}p_{m}%
	(t_{1})q_{\ell-m+1}(t_{2})\\
	\xi_{\ell}\left(  t_{1},t_{2}\right)   &  =\sum_{m=1}^{\kappa}r_{m}%
	(t_{1})s_{\ell-m+1}(t_{2}).
	\end{align*}
	Now, according to Lemma \ref{lemma:eta_to_integral} we can replace each term
	$\bar{\eta}^{\pm}\left(  p,q,r,s\right)  $ by the corresponding $\eta^{\pm
	}\left(  p,q,r,s\right)  $ up to an error of order $O(M^{-2})$. Therefore, it
	suffices to prove that the right hand side of (\ref{eq:eta_asym_derivs}) is
	equal to $\sum_{\ell=1}^{2\kappa-1}\phi_{\ell}+O(M^{-1})$. But this follows
	directly from the fact that
	\begin{align*}
	\mu_{\ell}\left(  T_{s},0\right)   &  =\left[  \mathcal{R}_{2}\left(
	p,q\right)  \right]  _{M,\ell}+O(M^{-1})\\
	\mu_{\ell}\left(  0,T_{s}\right)   &  =\left[  \mathcal{R}_{1}\left(
	p,q\right)  \right]  _{1,\ell}+O(M^{-1})\\
	\mu_{\ell}\left(  T_{s}/2,T_{s}/2\right)   &  =\left[  \mathcal{R}_{1}\left(
	p,q\right)  \right]  _{M,\ell}+O(M^{-1})\\
	&=\left[  \mathcal{R}_{2}\left(
	p,q\right)  \right]  _{1,\ell}+O(M^{-1})\\
	\mu_{\ell}\left(  0,T_{s}/2\right)   &  =\left[  \mathcal{S}_{2}\left(
	p,q\right)  \right]  _{1,\ell}+O(M^{-1})\\
	\mu_{\ell}\left(  T_{s}/2,0\right)   &  =\left[  \mathcal{S}_{2}\left(
	p,q\right)  \right]  _{M,\ell}+O(M^{-1})\\
	\mu_{\ell}\left(  T_{s},T_{s}/2\right)   &  =\left[  \mathcal{S}_{1}\left(
	p,q\right)  \right]  _{M,\ell}+O(M^{-1})\\
	\mu_{\ell}\left(  T_{s}/2,T_{s}\right)   &  =\left[  \mathcal{S}_{1}\left(
	p,q\right)  \right]  _{1,\ell}+O(M^{-1})
	\end{align*}
	and equivalently for $\xi_{\ell}$, replacing $p,q$ with $r,s$. This concludes
	the proof of the proposition.
\end{proof}

The application of the above proposition may prove to be difficult due the
presence of the term on the right hand side of (\ref{eq:eta_asym_derivs}),
which is difficult to interpret. The following corollary establishes that
under PR conditions, this term is zero.

\begin{corollary}
	\label{corollary:etas_asym_deriv_PR}Under the above assumptions and
	definitions, if $\left(  r,s\right)  $ are PR-compliant and $p,q$ are
	symmetric or anti-symmetric but have the opposite symmetry, we can write
	\begin{align*}
	&\eta^{+}\left(  p^{\prime},q,r,s\right)  -\eta^{+}\left(  p,q^{\prime
	},r,s\right)  +\eta^{-}\left(  p,q,r^{\prime},s\right)-\eta^{-}\left(p,q,r,s^{\prime}\right)=O(M^{-1}).
	\end{align*}
	
\end{corollary}

\begin{proof}
	It follows from the PR conditions that $\mat{U}^{-}\mathcal{S}\left(
	r,s\right)  $ is an all-zero matrix, whereas $\mat{U}^{+}\mathcal{R}\left(
	r,s\right)  $ contains zeros everywhere except for the central column, which
	is filled with $1$s. Proposition \ref{proposition:etas_asym_deriv} therefore
	establishes that%
	\begin{align*}
	 \eta^{+}\left(  p^{\prime},q,r,s\right)  -\eta^{+}\left(  p,q^{\prime
	},r,s\right)  +\eta^{-}\left(  p,q,r^{\prime},s\right)-\eta^{-}\left(
	p,q,r,s^{\prime}\right) &=\left[  \mathcal{R}\left(  p,q\right)  \right]  _{2M,\kappa}-\left[
	\mathcal{R}\left(  p,q\right)  \right]  _{1,\kappa}\\
	&+\left[  \mathcal{R}\left(
	p,q\right)  \right]  _{M,\kappa}-\left[  \mathcal{R}\left(  p,q\right)
	\right]  _{M+1,\kappa}+O(M^{-1})
	\end{align*}
	and the result follows from symmetry.
\end{proof}

Before we conclude this appendix, we introduce another asymptotic result that
will prove useful in the situation where two of the pulses meet the PR conditions.

\begin{lemma}
	\label{lemma:etas_asym_deriv_PR}Under the above definitions and hypotheses,
	assume additionally that $p$ and $q$ are perfect reconstruction pulses and
	that the analog waveforms $r,s$ and $r^{\prime},s^{\prime}$ are zero at the
	extreme of their support. Then, we can write%
	\[
	\eta^{\pm}\left(  p,q,r^{\prime},s\right)  -\eta^{\pm}\left(  p,q,r,s^{\prime
	}\right)  =O(M^{-2})
	\]
	
\end{lemma}

\begin{proof}
	We know that $\mat{U}^{-}\mathcal{S}\left(  p,q\right)  $ is an all zero
	matrix whereas the entries of $\mat{U}^{+}\mathcal{R}\left(  p,q\right)  $
	are all zero except for the central column, which is filled with ones. This
	means that tr$\left[  \mat{U}^{\mathbf{+}}\mathcal{S}(p,q)\mathcal{S}%
	(r,s)^{T}\right]  =0$ and
	\begin{align*}
	\frac{1}{2M}\text{tr}\left[  \mat{U}^{\mathbf{+}}\mathcal{R}%
	(p,q)\mathcal{R}(r,s)^{T}\right]& =\frac{1}{2M}\sum_{n=1}^{2M\kappa
	}r[n]s[2M\kappa-n+1]\\
	& \hskip2em =\frac{1}{T_{s}}\int_{0}^{\kappa T_{s}}r(t)s(\kappa T_{s}-t)dt+O(M^{-2})
	\end{align*}
	where the last identity follows from the Riemann integral definition.
	Therefore, since
	\begin{align*}
	&\frac{1}{T_{s}}\int_{0}^{\kappa T_{s}}r^{\prime}(t)s(\kappa T_{s}%
	-t)dt-\frac{1}{T_{s}}\int_{0}^{\kappa T_{s}}r(t)s^{\prime}(\kappa
	T_{s}-t)dt=r(\kappa T_{s})s(0)-r(\kappa T_{s})s(0)=0,
	\end{align*}

	we obtain the result.
\end{proof}

% use section* for acknowledgment
%\section*{Acknowledgment}
%
%
%The research reported herein was partly funded by Fonds pour la Formation \`a la Recherche dans l'Industrie et dans l'Agriculture (F.R.I.A.).

% Can use something like this to put references on a page
% by themselves when using endfloat and the captionsoff option.
\ifCLASSOPTIONcaptionsoff
  \newpage
\fi

% trigger a \newpage just before the given reference
% number - used to balance the columns on the last page
% adjust value as needed - may need to be readjusted if
% the document is modified later
%\IEEEtriggeratref{8}
% The "triggered" command can be changed if desired:
%\IEEEtriggercmd{\enlargethispage{-5in}}

% references section

% can use a bibliography generated by BibTeX as a .bbl file
% BibTeX documentation can be easily obtained at:
% http://www.ctan.org/tex-archive/biblio/bibtex/contrib/doc/
% The IEEEtran BibTeX style support page is at:
% http://www.michaelshell.org/tex/ieeetran/bibtex/
\bibliographystyle{IEEEtran}

\bibliography{IEEEabrv,refs}
\end{document}